\documentclass[11pt]{article}

\usepackage[utf8]{inputenc}
\usepackage{amsmath, amssymb, amsthm}
\usepackage{mathtools}
\usepackage{authblk}
\usepackage[backend=biber, style=alphabetic, maxbibnames=99]{biblatex}
\bibliography{refs}
\usepackage[sc]{mathpazo}
\usepackage{bbm}
\usepackage{dsfont}
\usepackage[basic]{complexity}
\usepackage{algorithm2e}
\usepackage{caption}
\usepackage{subcaption}
\usepackage{pgfplots}
\usepackage{tikz}
\usepackage{fullpage}
\usepackage{parskip}
\usetikzlibrary{shapes, patterns, decorations, fit, intersections}
\usepgfplotslibrary{patchplots}
\usepackage{ifthen}
\usepackage{hyperref}
\usepackage{cleveref}

\begingroup
    \makeatletter
    \@for\theoremstyle:=definition,remark,plain\do{%
        \expandafter\g@addto@macro\csname th@\theoremstyle\endcsname{%
            \addtolength\thm@preskip\parskip
            }%
        }
\endgroup

\let\R\relax
\newcommand*{\R}{\mathbb{R}}
\newcommand*{\Z}{\mathbb{Z}}
\newcommand*{\1}{\mathds{1}}

\let\span\relax
\DeclareMathOperator{\span}{span}
\DeclareMathOperator{\conv}{conv}
\DeclareMathOperator{\PM}{PM}

\DeclareMathOperator{\avg}{avg}
\DeclareMathOperator*{\argmin}{argmin}

\DeclarePairedDelimiter{\norm}{\lVert}{\rVert}
\DeclarePairedDelimiter{\card}{\lvert}{\rvert}
\DeclarePairedDelimiter{\abs}{\lvert}{\rvert}
\providecommand{\given}{}
\DeclarePairedDelimiterX{\set}[1]\{\}{\renewcommand\given{\nonscript\:\delimsize\vert\nonscript\:\mathopen{}}#1}
\let\Pr\relax
\DeclarePairedDelimiterXPP{\Pr}[1]{\mathbb{P}}[]{}{\renewcommand\given{\nonscript\:\delimsize\vert\nonscript\:\mathopen{}}#1}
\DeclarePairedDelimiterXPP{\Ex}[1]{\mathbb{E}}[]{}{\renewcommand\given{\nonscript\:\delimsize\vert\nonscript\:\mathopen{}}#1}
\DeclarePairedDelimiter{\dotp}{\langle}{\rangle}

\newclass{\quasiNC}{quasi\text{-}NC}
\newcommand*{\ds}[2][S]{\dotp{\1_{\delta(#1)}, #2}}
\newcommand*{\dv}[1]{\ds[v]{#1}}
\newcommand*{\eqtwo}{\stackrel{2}{\equiv}}
\newcommand*{\bcard}[1]{\overline{[#1]}}
\newcommand*{\cF}{\mathcal{F}}

\theoremstyle{plain}
\newtheorem{theorem}{Theorem}
\newtheorem{lemma}{Lemma}
\newtheorem{corollary}{Corollary}
\theoremstyle{definition}
\newtheorem{definition}{Definition}
\theoremstyle{remark}
\newtheorem{remark}{Remark}

\SetKwInOut{Input}{Input}
\SetKwInOut{Output}{Output}
\SetKwBlock{InParallel}{in parallel do}{end}
\SetKwFor{ParallelFor}{for}{in parallel do}{end}
\SetKwFunction{PerfectMatching}{\textsc{PerfectMatching}}
\SetKwFunction{BalancedViableSet}{\textsc{BalancedViableSet}}
\SetKwFunction{Preprocess}{\textsc{Preprocess}}
\SetKwFunction{Reduce}{\textsc{Reduce}}
\SetKwFunction{DisjointOddSets}{\textsc{DisjointOddSets}}
\SetKwFunction{MakeConnected}{\textsc{MakeConnected}}
\SetKwFunction{ShrinkDegreeTwos}{\textsc{ShrinkDegreeTwos}}
\SetKwFunction{Uncross}{\textsc{Uncross}}
\SetKwFunction{MergeUncross}{\textsc{MergeUncross}}
\SetKwFunction{ConnectedComponents}{\textsc{ConnectedComponents}}
\SetKwFunction{CheckBalancedViable}{\textsc{CheckBalancedViable}}

\provideboolean{havetikz}
\setboolean{havetikz}{true}
\newcommand*{\Pic}[1]{\ifhavetikz\tikz \pic{#1};\fi}

\definecolor{Black}{RGB}{0, 0, 0}
\definecolor{LightGray}{RGB}{216, 216, 216}
\definecolor{Gray}{RGB}{127, 127, 127}
\definecolor{Orange}{RGB}{237, 125, 49}
\definecolor{LightOrange}{RGB}{251, 229, 214}
\definecolor{Yellow}{RGB}{255, 192, 0}
\definecolor{Blue}{RGB}{91, 155, 213}
\definecolor{LightBlue}{RGB}{222, 235, 247}
\definecolor{Green}{RGB}{112, 173, 71}
\definecolor{LightGreen}{RGB}{226, 240, 217}
\definecolor{Navy}{RGB}{68, 114, 196}
\definecolor{LightNavy}{RGB}{218, 227, 243}
\newcommand*{\Orange}[1]{\textcolor{Orange}{#1}}

\tikzstyle{node}=[circle, line width=1, draw=Black, fill=Gray, inner sep=3]
\tikzstyle{edge}=[line width=1, color=Gray]
\tikzstyle{highlight}=[line width=2, color=Orange]
\tikzstyle{set}=[line width=1, color=Black, dashed]
\tikzstyle{fillset}=[color=Black, fill=Gray, line width=1]
\makeatletter
\tikzset{use path/.code=\tikz@addmode{\pgfsyssoftpath@setcurrentpath#1}}
\makeatother

\tikzset{
	graph/.pic={
		\foreach \v/\coord in {{a/(-2, -1)}, {b/(-2, 1)}, {c/(-1, 0)}, {d/(2, -1)}, {e/(2, 1)}, {f/(1, 0)}} \node[node] (\v) at \coord {};
		\foreach \u/\v in {a/d, b/e, c/f, a/b, b/c, c/a, d/e, e/f, f/d} \draw[edge] (\u) -- (\v);
		\draw[set] (-3, 0) to[out=90, in=180] (-2, 1.5) to[out=0, in=90] (-0.5, 0) to[out=-90, in=0] (-2, -1.5) to[out=180, in=-90] cycle;
		\foreach \u/\v/\pos in {b/e/above, b/c/above, e/f/above, a/b/left, d/e/right} \draw (\u) edge[edge] node[pos=0.5, \pos, Black] {$\frac{1}{3}$} (\v);
		\foreach \u/\v/\pos in {a/d/below, c/f/above}
			\path (\u) edge[highlight] node[pos=0.5, \pos, Black] {$\frac{1}{3}\Orange{-\epsilon}$} (\v);
		\foreach \u/\v/\pos in {a/c/right, d/f/left}
			\path (\u) edge[highlight] node[pos=0.5, \pos, Black] {$\frac{1}{3}\Orange{+\epsilon}$} (\v);
		\node at (-3, 1.5) {$S$};
	}
}

\tikzset{
	contracted/.pic={
		\foreach \v/\coord in {{a/(-2, 0)}, {b/(-2, 0)}, {c/(-2, 0)}, {d/(2, -1)}, {e/(2, 1)}, {f/(1, 0)}} \node[node] (\v) at \coord {};
		\foreach \u/\v in {a/d, b/e, c/f, a/b, b/c, c/a, d/e, e/f, f/d} \draw[edge] (\u) -- (\v);
		\foreach \u/\v/\pos in {a/d/below, a/e/above, a/f/right, d/e/right, e/f/above, f/d/above}
			\path (\u) edge[edge] node[pos=0.5, Black, \pos] {$\frac{1}{3}$} (\v);
		\node at (-2.3, 0.3) {$S$};
	}
}

\tikzset{
	matching-polytope-base/.pic={
		\begin{scope}[scale=2]
			\coordinate (O) at (0.5, 0.5, 0.5);
			\coordinate (U1) at (1, 0.5, 0.5);
			\coordinate (U2) at (0.5, 1, 0.5);
			\coordinate (V) at (1, 1, 0.5);
			\coordinate (A1) at (0, 0, 0);
			\coordinate (A2) at (1, 0, 0);
			\coordinate (A3) at (1, 1, 0);
			\coordinate (A4) at (0, 1, 0);
			\coordinate (B1) at (0, 0, 1);
			\coordinate (B2) at (1, 0, 1);
			\coordinate (B3) at (1, 1, 1);
			\coordinate (B4) at (0, 1, 1);
			\coordinate (C1) at (1, 0.6, 0);
			\coordinate (C2) at (0.6, 1, 0);
			\coordinate (D1) at (1, 0.6, 1);
			\coordinate (D2) at (0.6, 1, 1);
			\coordinate (E1) at (-0.5, 1.5, -0.5);
			\coordinate (E2) at (-0.5, 1.5, 1.5);
			\coordinate (F2) at (-0.5, 1.5, 1);
			\coordinate (E4) at (1.5, -0.5, -0.5);
			\coordinate (E3) at (1.5, -0.5, 1.5);
			\coordinate (F3) at (1.5, -0.5, 1);
			\coordinate (M1) at (0, 0, 0.5);
			\coordinate (M2) at (-0.5, -0.5, 0.5);
		\end{scope}
	}
}

\tikzset{
	bipartite-matching-polytope0/.pic={
		\pic{matching-polytope-base};
		\begin{scope}[fill=Gray, opacity=0.3]
			\fill (A1) -- (A2) -- (A3) -- (A4) -- cycle;
			\fill (A1) -- (A2) -- (B2) -- (B1) -- cycle;
			\fill (A1) -- (A4) -- (B4) -- (B1) -- cycle;
		\end{scope}
	}
}

\tikzset{
	bipartite-matching-polytope1/.pic={
		\begin{scope}[line width=1, Black, dashed, opacity=0.6]
			\foreach \u/\v in {A1/A4, A1/A2, A1/B1}
				\draw (\u) -- (\v);
		\end{scope}
	}
}
\tikzset{
	bipartite-matching-polytope2/.pic={
		\begin{scope}[fill=Gray, opacity=.3]
			\fill (B1) -- (B2) -- (B3) -- (B4) -- cycle;
			\fill (B3) -- (A3) -- (A4) -- (B4) -- cycle;
			\fill (B3) -- (A3) -- (A2) -- (B2) -- cycle;
		\end{scope}
		\begin{scope}[line width=1, Black, opacity=1.0]
			\foreach \u/\v in {A2/A3, A3/A4, B1/B2, B2/A2, B1/B4, B4/A4, B3/B2, B3/B4, B3/A3}
				\draw (\u) -- (\v);
		\end{scope}
	}
}

\tikzset{
	bipartite-matching-polytope/.pic={
		\pic{bipartite-matching-polytope0};
		\pic{bipartite-matching-polytope1};
		\pic{bipartite-matching-polytope2};
	}
}

\tikzset{
	bipartite-matching-polytope01/.pic={
		\pic{bipartite-matching-polytope0};
		\pic{bipartite-matching-polytope1};
	}
}

\tikzset{
	matching-polytope0/.pic={
		\pic{matching-polytope-base};
		\begin{scope}[fill=Gray, opacity=0.3]
			\fill (A1) -- (A2) -- (C1) -- (C2) -- (A4) -- cycle;
			\fill (A1) -- (A2) -- (B2) -- (B1) -- cycle;
			\fill (A1) -- (A4) -- (B4) -- (B1) -- cycle;
		\end{scope}
	}
}

\tikzset{
	matching-polytope1/.pic={
		\begin{scope}[line width=1, Black, dashed, opacity=0.6]
			\foreach \u/\v in {A1/A4, A1/A2, A1/B1}
				\draw (\u) -- (\v);
		\end{scope}
	}
}

\tikzset{
	matching-polytope2/.pic={
		\begin{scope}[fill=Gray, opacity=.3]
			\fill (B1) -- (B2) -- (D1) -- (D2) -- (B4) -- cycle;
			\fill (B2) -- (D1) -- (C1) -- (A2) -- cycle;
			\fill (B4) -- (D2) -- (C2) -- (A4) -- cycle;
			\fill (C1) -- (D1) -- (D2) -- (C2) -- cycle;
		\end{scope}
		\begin{scope}[line width=1, Black, opacity=1.0]
			\foreach \u/\v in {B1/B2, B2/A2, B1/B4, B4/A4, A2/C1, C1/C2, C2/A4, B2/D1, D1/D2, D2/B4, C1/D1, C2/D2}
				\draw (\u) -- (\v);
		\end{scope}
	}
}

\tikzset{
	matching-polytope/.pic={
		\pic{matching-polytope0};
		\pic{matching-polytope1};
		\pic{matching-polytope2};
	}
}

\tikzset{
	matching-polytope01/.pic={
		\pic{matching-polytope0};
		\pic{matching-polytope1};
	}
}

\tikzset{
	bipartite-matching-polytope-seq/.pic={
		\pic at (-2.5, 0) {bipartite-matching-polytope01};
		\node at (O) {\LARGE $\Orange{\bullet}$};
		\draw[Yellow, line width=2, ->] (O) -- (U1);
		\pic at (-2.5, 0) {bipartite-matching-polytope2};
		\node at (0, 0.5) {\Huge $+$};
		\pic at (1.5, 0) {bipartite-matching-polytope01};
		\node at (O) {\LARGE $\Orange{\bullet}$};
		\draw[Yellow, line width=2, ->] (O) -- (U2);
		\pic at (1.5, 0) {bipartite-matching-polytope2};
		\node at (0, -1) {$\underbrace{\hspace{15em}}$};
		\node at (0, -1.6) {\Huge $\Downarrow$};
		\pic at (-0.5, -4.2) {bipartite-matching-polytope01};
		\node at (O) {\LARGE $\Orange{\bullet}$};
		\draw[Yellow, line width=2] (O) -- (barycentric cs:C1=1,D1=1,C2=1,D2=1);
		\draw[Yellow, line width=2, ->] (barycentric cs:C1=1,D1=1,C2=1,D2=1) -- (V);
		\pic at (-0.5, -4.2) {bipartite-matching-polytope2};
	}
}

\tikzset{
	matching-polytope-seq/.pic={
		\pic at (-2.5, 0) {matching-polytope01};
		\node at (O) {\LARGE $\Orange{\bullet}$};
		\draw[Yellow, line width=2, ->] (O) -- (U1);
		\pic at (-2.5, 0) {matching-polytope2};
		\node at (0, 0.5) {\Huge $+$};
		\pic at (1.5, 0) {matching-polytope01};
		\node at (O) {\LARGE $\Orange{\bullet}$};
		\draw[Yellow, line width=2, ->] (O) -- (U2);
		\pic at (1.5, 0) {matching-polytope2};
		\node at (0, -1) {$\underbrace{\hspace{15em}}$};
		\node at (0, -1.6) {\Huge $\Downarrow$};
		\pic at (-0.5, -4.2) {matching-polytope01};
		\node at (O) {\LARGE $\Orange{\bullet}$};
		\draw[Yellow, line width=2] (O) -- (barycentric cs:C1=1,D1=1,C2=1,D2=1);
		\pic at (-0.5, -4.2) {matching-polytope2};
		\draw[Yellow, line width=2, ->] (barycentric cs:C1=1,D1=1,C2=1,D2=1) -- (V);
	}
}

\tikzset{
	hyperplane/.pic={
		\pic{matching-polytope-base};
		\fill[Blue, opacity=0.9] (B4) -- (A4) -- (A2) -- (B2) -- (F3) -- (E4) -- (E1) -- (F2) -- cycle;
		\draw[Orange, line width=2, ->] (M1) -- (M2) node[Black, pos=1.2] {$w$};
		
		\pic{matching-polytope0};
		\pic{matching-polytope1};
		\draw[Orange, line width=2] (O) -- (M1);
		\fill[Blue, opacity=0.9] (A4) -- (B4) -- (B2) -- (A2) -- cycle;
		\node at (O) {\LARGE $\Orange{\bullet}$};
		\draw[Yellow, line width=2, ->] (O) -- (U1);
		\draw[Yellow, line width=2, ->] (O) -- (U2);
		\pic{matching-polytope2};
		\fill[Blue, opacity=0.9] (E2) -- (F2) -- (F3) -- (E3) -- cycle;
		\node[rotate=45] at (barycentric cs:E1=7,E2=7,E3=1,E4=1) {$\dotp{w, x}=\text{const}$};
		\node (minimizer) at (barycentric cs:C1=1,C2=1,D1=1,D2=1) {\LARGE $\bullet$};
		\node (text) at (barycentric cs:O=1,V=-2) {minimizer of $\dotp{w,x}$};
		\draw[dashed, ->] (text) edge[out=-120, in=30] (minimizer);
	}
}

\tikzset{
	shrinking1/.pic={
		\foreach \u/\coord in {{a/(0, 0)}, {b/(2, 0)}, {c/(1, 0.5)}, {d/(0, -2)}, {e/(-0.5, -1)}, {f/(-1, -2)}, {g/(1, -2)}, {h/(2, -1)}, {i/(1, -1)}, {j/(2.3, -2.3)}}
			\node[node] (\u) at \coord {\u};
		\foreach \u/\v in {a/b, b/c, c/d, d/e, e/f, f/d, d/g, g/h, h/i, i/a, a/e, h/b, i/g, g/j, h/j}
			\draw[edge] (\u) -- (\v);
		\draw[set] (-1.5, -2) to[out=90, in=-180] (-0.5, -0.5) to[out=0, in=90] (0.5, -2) to[out=-90, in=-90] cycle;
	}
}

\tikzset{
	shrinking2/.pic={
		\foreach \u/\coord in {{a/(0, 0)}, {b/(2, 0)}, {c/(1, 0.5)}, {def/(-0.8, -1.3)}, {g/(1, -2)}, {h/(2, -1)}, {i/(1, -1)}, {j/(2.3, -2.3)}}
			\node[node] (\u) at \coord {\u};
		\foreach \u/\v in {a/b, b/c, c/def, def/g, g/h, h/i, i/a, a/def, h/b, i/g, g/j, h/j}
			\draw[edge] (\u) -- (\v);
		\draw[set] (0.7, -0.7) to[out=45, in=60] (2.5, -1) to[out=-120, in=30] (1, -2.5) to[out=-150, in=-135] cycle;
		\draw[set] (-1.2, -2.3) to[out=-145, in=-135] (-0.4, 0.5) to[out=45, in=50] (1.6, 0.5) to[out=-130, in=35] cycle;
	}
}

\tikzset{
	shrinking3/.pic={
		\foreach \u/\coord in {{acdef/(0, -0.5)}}
			\node[node, color=Orange, fill=LightOrange] (\u) at \coord {\u};
		\foreach \u/\coord in {{b/(2, 0)},  {ghi/(1.2, -1.2)}, {j/(2.3, -2.3)}}
			\node[node] (\u) at \coord {\u};
		\foreach \u/\v in {acdef/b, acdef/ghi, ghi/b, ghi/j}
			\draw[edge] (\u) -- (\v);
	}
}

\tikzset{
	shrinking/.pic={
		\pic{shrinking1};
		\node at (3, -1) {\Huge $\Rightarrow$};
		\pic at (5, 0) {shrinking2};
		\node at (8.2, -1) {\Huge $\Rightarrow$};
		\pic at (9.5, 0) {shrinking3};
	}
}

\tikzset{
	sameblock/.pic={
		\node[node] (a1) at (-1, -1) {};
		\node[node] (a2) at (1, -1) {};
		\node[node] (b1) at (-1, 1) {};
		\node[node] (b2) at (1, 1) {};
		\coordinate (a0) at (-1, -2); 
		\coordinate (a3) at (1, -2);
		\coordinate (b0) at (-1, 2);
		\coordinate (b3) at (1, 2);
		\draw (a0) edge[highlight] node[pos=0.5, Black, right] {$-\epsilon_1$} (a1);
		\path (a1) edge[highlight] node[pos=0.5, Black, above] {$+\epsilon_1$} (a2);
		\path (a2) edge[highlight] node[pos=0.5, Black, left] {$-\epsilon_1$} (a3);
		\path (a3) edge[dashed, out=-120, in=-60] (a0);
		\path (b0) edge[highlight] node[pos=0.5, Black, right] {$-\epsilon_2$} (b1);
		\path (b1) edge[highlight] node[pos=0.5, Black, below] {$+\epsilon_2$} (b2);
		\path (b2) edge[highlight] node[pos=0.5, Black, left] {$-\epsilon_2$} (b3);
		\path (b3) edge[dashed, out=120, in=60] (b0);
		\draw[set] (-2, 0) to[out=90, in=180] (0, 2) to[out=0, in=90] (2, 0) to[out=-90, in=0] (0, -2) to[out=180, in=-90] cycle;
		\node at (0, 0) {$S$};
	}
}

\title{Planar Graph Perfect Matching is in NC}
\author[1]{Nima Anari}
\author[2]{Vijay V.~Vazirani}

\affil[1]{Computer Science Department, Stanford University,\footnote{Part of this work was done while the first author was visiting the Simons Institute for the Theory of Computing. It was partially supported by the DIMACS/Simons Collaboration on Bridging Continuous and Discrete Optimization through NSF grant CCF-1740425.} \textsf{anari@cs.stanford.edu}}
\affil[2]{Computer Science Department, University of California, Irvine,\footnote{Supported in part by NSF grant CCF-1216019.} \textsf{vazirani@ics.uci.edu}}

\date{}

\begin{document}
	\maketitle
	\begin{abstract}
Is perfect matching in \NC? That is, is there a deterministic fast parallel algorithm for it? This has been an outstanding
open question in theoretical computer science for over three decades, ever since the discovery of \RNC~matching algorithms.
Within this question, the case of planar graphs has remained an enigma: On the one hand, counting the
number of perfect matchings is far harder than finding one (the former is \#\P-complete and the latter is in \P),
and on the other, for planar graphs, counting has long been known to be in \NC~whereas finding one has
resisted a solution.

In this paper, we give an \NC~algorithm for finding a perfect matching in a planar graph.
Our algorithm uses the above-stated fact about counting matchings in a crucial way. Our main new
idea is an \NC~algorithm for finding a face of the perfect matching 
polytope at which $\Omega(n)$ new conditions, involving constraints of the polytope,
are simultaneously satisfied. Several other ideas are also
needed, such as finding a point in the interior of the minimum weight face of this polytope and
finding a balanced tight odd set in \NC.
\end{abstract}

\section{Introduction}

Is perfect matching in \NC? That is, is there a deterministic parallel algorithm that computes a perfect matching in a graph
in polylogarithmic time using polynomially many processors? This has been an outstanding
open question in theoretical computer science for over three decades, ever since the discovery of \RNC~matching algorithms \cite{KUW, MVV}.
Within this question, the case of planar graphs had remained an enigma: For general graphs, counting the number of perfect matchings 
is far harder than finding one: the former is \#\P-complete \cite{Valiant} and the latter is in \P~\cite{Edmonds}.
However, for planar graphs, a polynomial time algorithm for counting perfect matchings 
was found by Kasteleyn, a physicist, in 1967 \cite{Kasteleyn}, and an \NC~algorithm follows easily\footnote{For a formal proof, in a slightly more general 
context, see \cite{K33}.},
given an \NC~algorithm for computing the determinant of a matrix, which was obtained by Csanky \cite{Csanky} in 1976.
On the other hand, an \NC~algorithm for finding a perfect matching in a planar graph has resisted a solution. In this paper, we provide such an algorithm.

An \RNC~algorithm for the decision problem, of determining if a graph has a perfect matching, was obtained by Lov{\'a}sz \cite{Lo},
using the Tutte matrix of the graph. The first \RNC~algorithm for the search problem, of actually finding a perfect matching, was
obtained by Karp, Upfal, and Wigderson \cite{KUW}. This was followed by a somewhat simpler algorithm due to Mulmuley, Vazirani,
and Vazirani \cite{MVV}. 

The matching problem occupies an especially distinguished position in the theory of algorithms: 
Some of the most central notions and powerful tools within this theory were discovered in the 
context of an algorithmic study of this problem, including the notion of polynomial time solvability 
\cite{Edmonds} and the counting class \#\P~\cite{Valiant}.
The parallel perspective has also led to such gains: The first \RNC~matching algorithm led to a fundamental understanding of the
computational relationship between search and decision problems \cite{KUW2} and the second algorithm yielded the Isolation
Lemma \cite{MVV}, which has found several applications in complexity theory and algorithms.
Considering the fundamental insights gained by an algorithmic study of
matching, the problem of obtaining an \NC~algorithm for it has remained a premier open question ever since the 1980s.

The first substantial progress on this question was made by Miller and Naor in 1989 \cite{MillerN}.
They gave an \NC~algorithm for finding a perfect matching in bipartite planar graphs using a flow-based approach.
In 2000, Mahajan and Varadarajan gave an elegant way of using the \NC~algorithm for counting perfect matchings to finding one, 
hence giving a different \NC~algorithm for bipartite planar graphs \cite{MahajanV}. Our algorithm is inspired by their approach.

In the last few years, several researchers have obtained \quasiNC~algorithms for matching and its generalizations; such algorithms
run in polylogarithmic time though they require $O(n^{\log^{O(1)} n})$ processors.
These algorithms achieve a partial derandomization of the Isolation Lemma for the specific problem 
addressed.
Several nice algorithmic ideas have been discovered in these works and our algorithm has benefited from some of these; in turn, it will not be surprising if some of our ideas turn out to be
useful in the resolution of the main open problem.
First, Fenner, Gurjar, and Thierauf gave a \quasiNC~algorithm for perfect matching in bipartite graphs \cite{FennerGT}, followed
by the algorithm of Svensson and Tarnawski for general graphs \cite{svensson2017matching}. Algorithms were also found for the generalization 
of bipartite matching to the linear matroid intersection problem by Gurjar and Thierauf \cite{GurjarT}, and to a further generalization 
of finding a vertex of a polytope with faces given by totally unimodular constraints, by Gurjar, Thierauf, and Vishnoi \cite{GurjarTV}.

Our main theorem is the following.
\begin{theorem}
\label{thm:main}
	There is an \NC~algorithm which given a planar graph, returns a perfect matching in it, if it has one.
\end{theorem}

In \cref{sec:extensions} we generalize \cref{thm:main} to finding a minimum
weight perfect matching if the edge weights are polynomially bounded and to finding a 
perfect matching in graphs of bounded genus; their common generalization easily follows.

	\section{Overview and Technical Ideas}

\begin{figure}
	\centering
	\begin{minipage}[b]{0.45\textwidth}
		\centering
		\Pic{graph}
		\captionof{figure}{Even cycle blocked by an odd set constraint. Example due to \cite{svensson2017matching, FennerGT}.}
		\label{fig:blocked}
	\end{minipage}%
	\hfill
	\begin{minipage}[b]{0.45\textwidth}
		\centering
		\Pic{contracted}
		\captionof{figure}{Resulting graph after shrinking the blocking tight odd set.}
		\label{fig:contracted}
	\end{minipage}
\end{figure}

\subsection{The Mahajan-Varadarajan algorithm and difficulties imposed by odd cuts}

We first give the idea behind the \NC~algorithm of Mahajan and Varadarajan for bipartite planar 
graphs. W.l.o.g.\ assume that the graph is matching covered, i.e., each edge is in a 
perfect matching. 
Using an oracle for counting the number of perfect matchings, they find a point $x$ in 
the interior of the perfect matching polytope and they show how to move this point to lower 
dimensional faces of the polytope until a vertex is reached; this will be a perfect matching. 
In a matching covered bipartite planar graph, every face (in a planar embedding) is the 
symmetric difference of two perfect matchings and modifying $x$ by increasing and decreasing 
alternate edges by the same (small) amount $\epsilon$ moves the point inside the polytope;
we will call this a {\em rotation} of the cycle.
Keep increasing $\epsilon$, starting from 0, until some edge $e$ on this cycle attains $x_e = 0$; in this case, $e$ is dropped. When this happens, 
$\epsilon$ cannot be increased anymore and we will say that the cycle is {\em blocked}. If so,
the point $x$ moves to a lower (by at least one) dimension face. To make substantial progress,
they observe that for any set of edge-disjoint faces, this process can be executed independently 
(by different amounts) in parallel on cycles corresponding to each of the faces, thereby reaching 
a face of the polytope of correspondingly lower dimension. Finally, they show how to find 
$\Omega(n)$ edge-disjoint faces in \NC, thereby terminating in $O(\log n)$  such iterations.


The fundamental difference between the perfect matching polytopes of bipartite and non-bipartite 
graphs is the additional constraint in the latter saying that each odd subset, $S$, of vertices must 
have a total of at least one edge in the cut $\delta(S)$ (see LP (\ref{eq:polytope}) in 
\cref{sec:polytope1}). This constraint introduces a second way in which
a cycle $C$ can be blocked, namely some odd set $S$, whose cut intersects $C$, may go tight
(\cref{sec:find-one}) and the $\epsilon$ introduced for this cycle cannot be increased any more 
(without moving the point outside the polytope and making it infeasible). If so the cycle will not 
lose an edge. However,
notice that since one of the odd set constraints has gone tight, we are already at a face of one lower
dimension! In this case, we will say that {\em cycle $C$ is blocked by odd set $S$}. For an example,
see \cref{fig:blocked} in which the point inside the polytope is the all $1/3$ vector
and highlighted cycle is blocked by odd set $S$. Observe that
the rotation chosen on the highlighted cycle leads to infeasibility on cut $S$.

How do we capitalize on this progress though? The obvious idea (which happens to need
substantially many additional ideas to get to an \NC~algorithm) is to shrink $S$, find a perfect 
matching in the shrunk graph, expand $S$, remove from it the vertex that is matched to a vertex
in $V-S$, and find a perfect matching on the remaining vertices of $S$. For an example of the shrunk graph, see \cref{fig:contracted}.
It is easy to see that at least one edge
of $C$ must have both endpoints in $S$, and therefore the shrunk graph is smaller than the 
original graph.

As stated above, a number of new ideas are needed to make this rough outline yield an \NC~algorithm.
First, a small hurdle: If $G$ is non-bipartite planar, the procedure of Mahajan and
Varadrajan will find $\Omega(n)$ edge-disjoint faces; however, not all of these faces may be even.
In fact, there are matching covered planar graphs having only one even face. To get around this,
we define the notion of an {\em even walk}: it consists of two odd faces with a path connecting
them; for convenience, we will call an even cycle an even walk as well (\cref{sec:independentset}).
We give an \NC~algorithm for finding $\Omega(n)$ edge-disjoint even walks in $G$ 
(\cref{sec:evenwalks}). Furthermore, it is easy to see that rotating an even walk also moves point 
$x$ inside the perfect matching polytope (this is done in \cref{lem:rotation}). 

\subsection{A key algorithmic issue and its resolution}

As in the case of a cycle, a walk is blocked either if it loses an edge or if an odd cut 
intersecting it goes tight. In either case, the point moves to a face of lower dimension. However, 
a new algorithmic question arises: In the first case, the amount of rotation required 
to make the walk lose an edge is easy to compute, similar to the bipartite case. But in the second
case, how do we find the smallest rotation so an odd cut intersecting the walk goes
just tight? Note that we seek the ``smallest'' rotation so no odd cut goes under-tight. 
We will postpone the answer to this question until we address the next hurdle.

\begin{figure}
	\centering
	\begin{minipage}[b]{0.45\textwidth}
		\centering
		\Pic{bipartite-matching-polytope-seq}
		\captionof{figure}{Parallel moves in the bipartite matching polytope.}
		\label{fig:bipartite-matching-polytope}
	\end{minipage}%
	\hfill
	\begin{minipage}[b]{0.45\textwidth}
		\centering
		\Pic{matching-polytope-seq}
		\captionof{figure}{Parallel moves in the non-bipartite matching polytope.}
		\label{fig:matching-polytope}
	\end{minipage}
\end{figure}

Next, we state a big hurdle: As in the bipartite case, to make substantial progress, we need to 
move the point
$x$ to a face of the polytope where each of the $\Omega(n)$ even walks is blocked.
Recall that in the bipartite case, we could modify all of the
edge-disjoint even cycles independently in parallel and, the resulting point still remains 
in the matching polytope.
However, in the non-bipartite case, executing these moves in parallel may take the point outside 
the polytope, i.e., it becomes infeasible. These two cases are illustrated in 
\cref{fig:bipartite-matching-polytope} and \cref{fig:matching-polytope}, respectively.
The reason for the latter is that whereas rotations on two different walks may be individually
feasible, executing them both may make an odd set go under-tight. This is made explicit
in \cref{fig:sameblock} in which the two walks can individually be rotated by $\epsilon_1$ and
$\epsilon_2$, respectively, without violating feasibility. However, executing them both 
simultaneously makes odd set $S$ go under-tight.

The following idea, which can be considered the main new idea in our work, helps us get around 
this hurdle: It suffices to find a weight function on edges, $w$, so that 
one of the half-spaces defined by $w$ contains the vector $w$ itself and the other contains, for
each of the even walks, the direction of motion resulting from its rotation. Then, the minimizer 
of $x\mapsto \dotp{w, x}$ in the polytope will lie in a face at which
each of the walks is blocked either because it has lost an edge or it intersects a tight odd cut. 
Furthermore, the minimizer is still a feasible point, so no odd cut goes under-tight. 
This idea is illustrated
in \cref{fig:hyperplane}. The two yellow arrows indicate independent moves on two edge-disjoint
walks which lead to two different faces of the polytope. Executing them both simultaneously
would take the point outside the polytope, as was illustrated in \cref{fig:matching-polytope}.
However, the minimizer of $\dotp{w, x}$ lies on a face of the polytope at which both
the walks are blocked.

The weight function $w$ is obtained as follows: The traversal of an even walk gives an ordered
list of edges, possibly with repetition, of even length. W.r.t.\ a weight function $w$, define 
the {\em circulation} of an even walk to be the difference of sum of weights 
of even and 
odd numbered edges in the traversal of this walk (\cref{sec:find-one}). We show that any weight
function $w$ that makes the circulation of each of the even walks non-zero suffices in the following
sense: Given such a function $w$, we can pick a direction of rotation for each of the even walks
so that one of the half-spaces defined by $w$ contains the vector $w$ and the other contains the 
direction of motion of each of the even walks (\cref{lem:main}). Moreover, such a function $w$
is easy to construct: in each walk, pick the weight of any one edge to be 1
and the rest 0 (\cref{sec:balanced-viable}).

\begin{figure}
	\centering
	\begin{minipage}[b]{0.45\textwidth}
		\centering
		\Pic{sameblock}
		\captionof{figure}{Odd set constraint violated when even walks are rotated independently.}
		\label{fig:sameblock}
	\end{minipage}%
	\hfill
	\begin{minipage}[b]{0.45\textwidth}
		\centering
		\Pic{hyperplane}
		\captionof{figure}{Minimizer of appropriate linear function $x\mapsto \dotp{w, x}$ blocks even walks.}
		\label{fig:hyperplane}
	\end{minipage}
\end{figure}

Next, we need to find the minimizer of $w$ in the polytope. For this, it suffices to construct 
an \NC~oracle for computing $\#G^e_w$, the
number of minimum weight matchings in $G$ containing the edge $e$. This oracle is constructed
by finding a Pfaffian orientation (\cref{sec:Tutte}) for $G$, appropriately substituting
for the variables in the Tutte matrix of $G$ and computing the determinant of the resulting matrix
(\cref{sec:avg}).

Some clarifications are due at this point: First, let us answer the opening question of this
section, i.e., how to rotate just one given walk so it gets blocked by one of the two ways? 
Interestingly enough, at present we know of no simpler method for one walk than for multiple walks 
(binary search comes to mind but that is not an elegant, analytic solution). 
Second, in the bipartite case,
either one of two measures of progress works when we simultaneously rotate many edge-disjoint 
cycles: the number of edges removed or the decrease in the dimension of the face we end up on.
In the non-bipartite case, if rotating a walk leads to an odd cut going tight, we will shrink
the cut, as stated above. As stated above, at least one edge of the walk will be in the odd
set and will get shrunk.
So, for the case of a single walk, both measures of progress stated above still work.
However, when we rotate $k$ edge-disjoint walks, say, and none of these walks loses an edge,
then the dimension of the face we end up on may not be smaller by $O(k)$.
The reason is that very few, or even one, odd cut may intersect each of the walks. However,
each walk will have at least one edge in a tight odd set, and therefore shrinking them will
result in $O(k)$ edges being shrunk. Hence, the first measure of progress still works.

\subsection{The rest of the ideas}

A number of ideas are still needed to get to an \NC~algorithm. First, for each walk that
does not lose an edge, we need to find a tight odd cut intersecting it. For this, we
use the result of Padberg and Rao \cite{PadbergR} stating that the Gomory-Hu tree of a graph 
will contain a minimum weight odd cut. We show how to find a Gomory-Hu tree in a weighted planar 
graph in \NC~(\cref{sec:GH}), a result of independent interest. 
Next, consider a walk $W$ which, w.r.t.\ the current point $x$ in
the polytope, is intersecting a tight odd cut. We rotate walk $W$ further slightly so the point $x$
becomes infeasible, i.e., $W$ now crosses an under-tight cut, or perhaps several of them
(\cref{lem:main}). Observe that rotating a walk leaves each singleton cut tight. Hence, w.r.t.\ $x$,
each minimum weight odd cut must be an under-tight cut that crosses $W$, and the
Gomory-Hu tree must contain one of them. Repeating this for all walks in parallel, we obtain
a set of tight odd cuts that intersect each of the walks.

However, these odd cuts cannot be shrunk
simultaneously because they may be crossing each other. We know there is a laminar family of tight odd
cuts, but how do we find it in \NC? We next give a divide-and-conquer-based procedure that finds 
the top level sets of one such laminar family; these top level sets can clearly be shrunk 
simultaneously (\cref{sec:uncross-all}). The procedure works as follows:
We partition the family of tight odd cuts into two (almost) equal subfamilies, recursively 
``uncross'' each subfamily to obtain its top level sets, and then merge these two into one family 
of top level sets. Clearly the last step is the crux of the procedure. The key to it lies in 
observing that two families 
of top level sets have a simple intersection structure, which can be exploited appropriately.

The proposed algorithm has now evolved to the following: shrink all top level sets and
recursively find a perfect matching in the shrunk graph, followed by recursively finding a perfect 
matching in each of the shrunk sets (after removing its matched vertex). This algorithm has 
polylog depth; however, it does not run in polylog time because of the following inherent 
sequentiality: matchings in the shrunk sets have to be found {\em after} finding a matching 
in the shrunk graph. The reason is that a matching in a shrunk 
set $S$ can be found only after knowing the vertex in $S$ that is matched outside $S$, and
this will be known only after finding a matching in the shrunk graph.
We next observe that if we could find a {\em balanced} tight odd cut,
we would be done by a simple divide-and-conquer strategy: match any edge in the cut and find
matchings in the two sides of the cut recursively, in parallel. The task of finding a
balanced tight odd cut is not straightforward though. It involves iteratively shrinking the top level
sets found, finding even walks and moving to the minimum weight face in the shrunk graph, etc.
(\cref{sec:balanced-viable}). This is illustrated in \cref{fig:shrinking}.
We show that $O(\log n)$ such iterations suffice for finding a balanced tight odd cut.

\begin{figure}
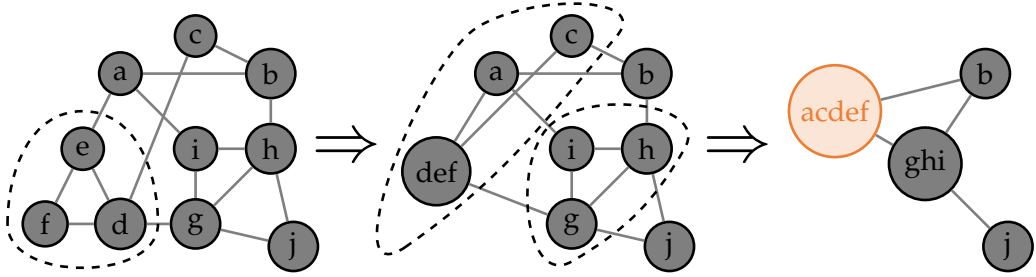

	\centering
	\Pic{shrinking}
	\captionof{figure}{Shrinking repeatedly yields a balanced viable set.}
	\label{fig:shrinking}
\end{figure}

	\section{Preliminaries}
\label{sec:prelims}

In this section, we will state several notions and algorithmic primitives we need for our \NC~algorithm for finding a perfect 
matching in a planar graph. 

\subsection{The Tutte matrix and Pfaffian orientations}
\label{sec:Tutte}

A key fact underlying our algorithm is that computing the number of perfect matchings in a planar graph lies in \NC. 
Let $G = (V, E)$ be an arbitrary graph (not necessarily planar).
Let $A$ be the symmetric adjacency matrix of $G$, i.e., corresponding to each edge $(i, j) \in E, \ A(i, j) = A(j, i) = 1$, and the
entries corresponding to non-edges are zero. Obtain matrix $T$ from $A$ by replacing for each edge $(i, j) \in E$, 
its two entries by $x_{ij}$ and $-x_{ij}$, so the entries below the diagonal are positive; clearly, $T$ is skew-symmetric.
$T$ is called the {\em Tutte matrix} for $G$. Its significance lies in that its determinant is non-zero as a polynomial iff $G$ has a 
perfect matching. However, computing this determinant is not easy: Simply writing it will require exponential space in general.

Next assume that $G$ has a perfect matching.
A simple cycle $C$ in $G$ is said to be {\em nice} if the removal of its vertices leave a graph having a perfect matching. If so, clearly, 
$C$ lies in the symmetric difference of two perfect matchings in $G$. Direct the edges of $G$ to obtain $\overrightarrow{G}$.
We will say that $\overrightarrow{G}$ is a {\em Pfaffian orientation} for $G$ if each nice cycle $C$ has an odd number of 
edges oriented in each way of traversing $C$. Its significance lies in the following: Let $(i, j) \in E$, with $i < j$. If in the 
Pfaffian orientation, this edge is directed from $i$ to $j$, then let $x_{ij} = 1$, otherwise let $x_{ij} = -1$.
Then the determinant of the resulting matrix is the square of the number of perfect matchings in $G$.

Of course, $G$ may not have a Pfaffian orientation. A key fact underlying our algorithm is that every planar graph has a
Pfaffian orientation and moreover, such an oriantation can be found in \NC~and the determinant can be computed in \NC~by Csanky's 
algorithm \cite{Csanky}. Hence we can answer the decision question of whether $G$ has a perfect matching in \NC.  


\subsection{The perfect matching polytope, its faces, and tight odd sets}
\label{sec:polytope1}

The perfect matching polytope for $G = (V, E)$ is defined in $\R^{E}$ and is given by the following set of linear equalities and inequalities \cite{Edmonds.matching}.

\begin{equation}
	\label{eq:polytope}
	\set*{x\in \R^E \given \begin{array}{ll}
		\dv{x}=1 & \forall v \in V,\\
		\ds{x}\geq 1 & \forall S\subset V,\text{ with }\card{S}\text{ odd},\\
		x_e \geq 0&\forall e \in E.
	\end{array}}
\end{equation}
We use the notation $\1_F$ to denote the indicator vector of a subset of edges $F\subseteq E$. For a subset of vertices $S$, we let $\delta(S)$ denote the edges that cross $S$, and by a slight abuse of notation we let $\delta(v)=\delta(\{v\})$ denote the set of edges adjacent to vertex $v$.

The perfect matching polytope is the convex hull of indicator vectors of all perfect matchings in $G$ and will be denoted by $\PM(G)$:

\[ \PM(G)=\conv\set*{ \1_{M}\mid M\text{ is a perfect matching of $G$}}. \]

For a given weight vector $w\in \R^E$ on edges, we can obtain minimum weight fractional and integral perfect matchings by
minimizing the linear function $x\mapsto\dotp{w, x}=\sum_e w_e x_e$ subject to the above-stated constraints. 
This set of fractional and integral perfect matchings form a face of $\PM(G)$ and will be denoted by $\PM(G, w)$.

One of the key steps needed by our algorithm is finding a point in the relative interior of the face $\PM(G, w)$ in \NC.
This requires computing a Pfaffian orientation for $G$ and then
evaluating the Tutte matrix for appropriate substitutions of the variables. The point we find will be exactly the average of the vertices, i.e., $\1_M$ for perfect matchings $M$, lying on the face $\PM(G, w)$. We denote this average by $\avg(\PM(G,w))$:
\[ \avg(\PM(G, w)) = \frac{1}{\card{\set{M\given \1_M\in \PM(G, w)}}}\sum_{M:\1_M\in \PM(G, W)} \1_M.\]

\begin{lemma}
	\label{lem:avg}
	Given a planar graph $G=(V,E)$ and an integral weight vector $w\in \Z^E$ represented in unary, there is an \NC~algorithm which returns $\avg(\PM(G,w))$.
\end{lemma}

We prove \cref{lem:avg} in \cref{sec:avg}.

In general, a face of $\PM(G)$ is defined by setting a particular set of inequalities to equalities. Let $\cal S$ be the family of
odd sets whose inequalities are set to equality. These will be called {\em tight odd sets}.
Two such tight odd sets $S_1, S_2 \in \cal S$ are said to {\em cross} if they are not disjoint and
neither is a subset of the other. If so, one can prove that either $S_1 \cap S_2$ and $S_1 \cup S_2$ are also tight odd sets or
$S_1 - S_2$ and $S_2 - S_1$ are tight odd sets.
In the former case one can remove the equality constraint for $S_1$ and replace it by the equality constraints for $S_1\cap S_2$ and $S_1\cup S_2$, and the face would not change. In the latter case $S_1$ can be replaced by $S_1-S_2$ and $S_2-S_1$ and still the face remains invariant. In either case, the new sets do not cross. The family $\cal S$ is said to be
{\em laminar} if no pair of sets in it cross. Given a family of tight odd sets $\cal S$, one can successively uncross pairs to
obtain a family of tight odd sets defining the same face of the polytope. This operation will result in a laminar family. However, for our purposes, we only need to work with the maximal sets in the laminar family. We define a similar notion of uncrossing for such top-level sets and show how they give us the space of equality constraints, by defining things appropriately.

\subsection{Finding maximal independent sets and even walks}
\label{sec:independentset}

One of the ingredients we use in multiple ways to design our algorithm is that a maximal independent set in a graph can be found in \NC.

\begin{lemma}[\cite{luby}]
	\label{lem:maximal}
	There is an \NC~algorithm for finding some maximal independent set in an input graph $G=(V,E)$.
\end{lemma}

Mahajan and Varadarajan used \cref{lem:maximal} to find linearly many edge-disjoint cycles in bipartite planar graphs \cite{MahajanV}. We use a similar step, but instead of cycles we have to work with even walks, i.e., cycles with possibly repeated edges.

\begin{definition}
	For this paper, an {\em even walk} is either a simple even length cycle in $G$ or the following structure:
	Let $C_1$ and $C_2$ be two odd length edge-disjoint cycles in $G$ and let $P$ be a path, edge-disjoint from $C_1,C_2$, connecting vertex $v_1$ of $C_1$
to vertex $v_2$ of $C_2$; if $v_1 = v_2$, $P$ will be the empty path. Starting from $v_1$, traverse $C_1$, then $P$ from $v_1$ 
to $v_2$, then traverse $C_2$, followed by $P$ from $v_2$ to $v_1$. This will be a walk that traverses an even number of edges 
and will also be called an even walk.
\end{definition}

Note that all of our walks start and end at the same location. We use \cref{lem:maximal} to derive the following. We prove \cref{lem:evenwalks} in \cref{sec:evenwalks}.

\begin{lemma}
	\label{lem:evenwalks}
	Suppose that $G=(V,E)$ is a connected planar graph with no vertices of degree $1$ and at most $\card{V}/2$ vertices of degree $2$. Then we can find $\Omega(\card{E})$ edge-disjoint even walks in $G$ by an \NC~algorithm.
\end{lemma}

	\section{Main Algorithm}
\label{sec:main-alg}
\subsection{Divide-and-conquer procedure}
\label{sec:algorithm}

In this section we will describe the algorithm we use to prove \cref{thm:main}. W.l.o.g.\ assume that the input graph has a perfect matching. We can easily check whether a perfect matching exists first, by counting the number of perfect matchings in \NC; see \cref{sec:Tutte}.

\begin{algorithm}
	\caption{Divide-and-conquer algorithm for finding a perfect matching.}
	\PerfectMatching{$G=(V,E)$}
	
	\eIf{$\card{V}=0$}{
		
		\Return $\emptyset$.
	}{
		Find a viable set $S$ with $\card{S}/\card{V}\in [c_1, 1-c_1]$.
		
		Let $w\leftarrow \1_{\delta(S)}$.
		
		Let $x\leftarrow \avg(\PM(G, w))$.
		
		Select an arbitrary edge $e\in \delta(S)$ with $x_e>0$.
		
		Let $G_1$ be the induced graph on $S$ with the endpoint of $e$ removed.
		
		Let $G_2$ be the induced graph on $V-S$ with the endpoint of $e$ removed.
		
		\InParallel{
			$M_1 \leftarrow$ \PerfectMatching{$G_1$}.
			
			$M_2 \leftarrow$ \PerfectMatching{$G_2$}.
		}
	}
	\Return $M_1\cup M_2\cup \set{e}$
	\bigskip
	\label{alg:perfectmatching}
\end{algorithm}

We use a divide-and-conquer approach. The pseudocode is given in \cref{alg:perfectmatching}. Given a graph $G=(V,E)$, our algorithm finds an odd set $S\subset V$, selects an edge $e\in \delta(S)$ as the first edge of the perfect matching, and then recursively extends this to a perfect matching in $S$ and $V-S$, without using any other edge of the cut $\delta(S)$.

Note that if $M$ is the output of our algorithm, by definition, $\card{M\cap \delta(S)}=1$. This prevents us from using an arbitrary odd set $S\subset V$ in the first step and motivates the following definition.
\begin{definition}
	Given a graph $G=(V, E)$, an odd set $S$ is called \emph{viable} if there exists at least one perfect matching $M\subseteq E$ with $\card{M\cap \delta(S)}=1$.
\end{definition}

In order for a step of the algorithm to make significant progress, i.e., reduce the size of the graph by a constant factor, we also require the viable set to be \emph{balanced}. That is, we require 
\[ c_1 \leq \frac{\card{S}}{\card{V}}\leq (1-c_1)\]
for some small constant $c_1>0$. Throughout the paper we will assume several constant upper bounds for $c_1$. At the end $c_1$ can be set to the lowest of these upper bounds.

Assuming that we are able to find a balanced viable set $S$ in \NC, we can prove \cref{thm:main}.
\begin{proof}[Proof of \cref{thm:main}]
	Since the set $S$ found by \cref{alg:perfectmatching} is feasible, there is at least one perfect matching $N$ with $\card{N\cap \delta(S)}=1$. On the other hand, for the weight vector $w=\1_{\delta(S)}$ and any perfect matching $N$, we have $\dotp{w, \1_N}=\card{N\cap \delta(S)}$, which is always at least one. So the minimum weight perfect matchings $N$ are exactly those that have a single edge in the cut $\delta(S)$. The point $x$ is the average of these perfect matchings, so for any edge $e$ with $x_e>0$, there is at least one minimum weight perfect matching $N\ni e$. This shows that $\set{e}$ can be extended to a perfect matching without using any other edge of $\delta(S)$ and therefore proves that $G_1$ and $G_2$ both have a perfect matching, an assumption we need to be able to recursively call the algorithm. This shows the correctness of the algorithm.
	
	We finish the proof by showing that the algorithm is in \NC. By \cref{lem:avg}, we can compute the point $x$ in \NC, and we assumed the viable set $S$ was found by an \NC~algorithm. So all of the steps of each recursive call can be executed in polylogarithmic time with a polynomially bounded number of processors. Notice that the recursion depth of the algorithm is at most $\log_{1/(1-c_1)}(\card{V})$  which is logarithmic in the input size. This is because the size of the graph gets reduced by a factor of $1-c_1$ in each recursive level. Since recursive calls are executed in parallel, this shows that the entire algorithm runs in polylog time.
\end{proof}

All that remains is finding a balanced viable set by an \NC~algorithm. This is done in \cref{sec:balanced-viable}.

\subsection{Finding a balanced viable set}
\label{sec:balanced-viable}

In this section we describe how to find a balanced viable set $S$ in a graph $G=(V,E)$ by an \NC~algorithm. Notice that a single vertex is, by definition, a viable set, but is not balanced unless $\card{V}\leq \frac{1}{c_1}$. So w.l.o.g.\ we can assume $\card{V}>\frac{1}{c_1}$.

The main idea behind our algorithm is the following: Suppose that we reduce the size of the graph $G$ by either removing edges not participating in perfect matchings from it, or shrinking tight odd sets (both w.r.t.\ some weight vector $w$). Any vertex in the shrunk graph corresponds to an odd set in the original graph $G$. This odd set is always viable. So if we manage to reduce the size of the shrunk graph enough so that it contains at most $1/c_1$ vertices, then the largest of the viable sets we get this way would have size at least $c_1\card{V}$. By being careful when we remove edges or shrink pieces, we can also make sure the size is not larger than $(1-c_1)\card{V}$; so the end result is a balanced viable set. See \cref{fig:shrinking} for a depiction.

The pseudocode is given in \cref{alg:balancedviableset}. Throughout the algorithm we maintain a mapping $f$ from the original vertices to the vertices of the current shrunk graph. We iteratively reduce the size of the graph by removing edges and/or contracting odd sets of vertices until one of the vertices contains a $c_1$ fraction of the original vertices. We then return the preimage of this vertex.

\begin{algorithm}
	\caption{Finding a balanced viable set.}
	
	\BalancedViableSet{$G_0=(V_0,E_0)$}
	
	Let $G=(V, E)$ be a copy of $G$ and let $f:V_0\to V$ be the identity map.
	
	\While{$\card{f^{-1}(v)}< c_1\card{V_0}$ for all $v\in V$}{
		$G, f\leftarrow \Preprocess(G, f)$.
		
		$G, f\leftarrow \Reduce(G, f)$.
%
%
%
%
%
	}
	
	Find $v\in V$ for which $\card{f^{-1}(v)}\geq c_1\card{V_0}$.
	
	\Return $f^{-1}(v)$.
	\bigskip
	\label{alg:balancedviableset}
\end{algorithm}

\begin{lemma}
	\label{lem:terminate}
	The while loop in \cref{alg:balancedviableset} finishes as soon as $\card{V}\leq \frac{1}{c_1}$.
\end{lemma}
\begin{proof}
	At any point in the algorithm, we have
	\[ \sum_{v\in V}\frac{\card{f^{-1}(v)}}{\card{V_0}} = 1. \]
	So when the number of terms in the sum is below $\frac{1}{c_1}$, one of them has to be larger than $c_1$.
\end{proof}

We further maintain the invariant that our graph $G$ is at all time planar, and has a perfect matching. This invariant is satisfied, because we restrict ourselves to manipulate $G$ in only one of the two following ways:
\begin{enumerate}
	\item We either remove an edge $e$ from $G$, where $e$ does not participate in any minimum weight perfect matching for some weight vector $w$.
	\item Or we shrink a set of vertices $S$ that is a tight odd set w.r.t.\ some point $x$ in the matching polytope.
\end{enumerate}

\begin{lemma}
	If a graph $G$ is planar and has a perfect matching, after the removal of an edge or shrinking of a set as described above, it continues to be planar and have a perfect matching.
\end{lemma}
\begin{proof}
	The lemma is obvious in the case of removing an edge. Planarity is automatically satisfied, and since the edge did not participate in any minimum weight perfect matching, the remaining graph still has a minimum weight perfect matching.
	
	In the case of shrinking a tight odd set, first note that the resulting graph would still have a perfect matching. Indeed, if we look at $x$ after shrinking $S$, it becomes a valid point in the matching polytope of the shrunk graph: The degree constraint of the shrunk vertex is satisfied because the odd set constraint of $S$ was originally tight.
	
	It remains to show why after shrinking $S$, the graph remains planar. If $S$ is internally connected, this follows from the fact that contracting edges preserves planarity. Otherwise, assume that $S=S_1\cup S_2$ where there are no edges between $S_1$ and $S_2$. Because $S$ is odd, either $S_1$ or $S_2$ must be odd, and the other even. W.l.o.g.\ assume that $S_1$ is odd. Note that $\ds[S_1]{x}\leq \ds{x}$ and equality holds if and only if $S_2$ does not have any outgoing edges. But $\ds{x}=1$, and $\ds[S_1]{x}\geq 1$, so there must be equality.
	
	So when $S$ was not internally connected, we showed that it must be a union of a connected component and a smaller tight odd set. By repeating this argument, we see that $S$ must always be a union of an internally connected odd set and a number of connected components. Now simply note that shrinking an entire connected component with some other vertex preserves planarity.
\end{proof}

It is also easy to see that any viable set in the resulting graph is a viable set in the original graph at any point, because any perfect matching in $G$ can be extended to a perfect matching in $G_0$. So we can return $f^{-1}(v)$ at any point if it is a balanced set.

The main loop in \cref{alg:balancedviableset} has two steps, \Preprocess and \Reduce. Although not explicitly stated in the pseudocode, at any point in the execution of either step, we can terminate the whole procedure by finding a balanced viable set and directly returning it.

First we preprocess the graph. Below we state the properties we expect to hold after preprocessing. We postpone the description of the procedure \Preprocess and the proof of \cref{lem:preprocess} to \cref{sec:preprocessing}.

\begin{lemma}
	\label{lem:preprocess}
	The procedure \Preprocess either finds a balanced viable set or after it returns the following conditions hold.
	\begin{enumerate}
		\item $G$ is connected.
		\item No vertex $v\in V$ has degree $1$ and at most half of the vertices have degree $2$.
		\item For all $v\in V$, we have $\card{f^{-1}(v)}<c_1 \card{V_0}$.
	\end{enumerate}
\end{lemma}

Now we describe the main step, i.e., \Reduce. The pseudocode is given in \cref{alg:reduce}.

Assuming \cref{lem:preprocess}, our goal is to either remove a constant fraction of the edges of $G$ or shrink pieces of $G$ so that a constant fraction of the edges get shrunk. The conditions satisfied after the preprocessing step, \cref{lem:preprocess}, ensure that we can apply \cref{lem:evenwalks} and find $\Omega(\card{E})$ edge-disjoint even walks, as we do in the first step of \cref{alg:reduce}.

\begin{algorithm}
	\caption{Reducing the size of a graph by shrinking vertices and removing edges.}
	
	\Reduce($G=(V, E)$, $f:V_0\to V$)
	
	Find $\Omega(\card{E})$ edge-disjoint even walks $W_1,\dots, W_k$.
	
	Let $w\leftarrow 0$, the zero weight vector.
	
	\ParallelFor{$W\in \set{W_1,\dots, W_k}$}{
		Set $w_e\leftarrow 1$ for the first edge $e$ of $W$.
	}
	
	Let $x\leftarrow \avg(\PM(G, w))$.
	
	\ParallelFor{$e\in E$ with $x_e=0$}{
		Remove edge $e$ from $G$.
	}
	
	Let $\mathcal{W}=\set{W_i \given W_i \text{ did not lose an edge}}$.
	
	Let $S_1,\dots,S_l\leftarrow \DisjointOddSets(G, f, x, \mathcal{W})$.
	
	Shrink each $S_i$ into a single vertex and update $f$ on $f^{-1}(S_i)$ to point to the new vertex.
	\bigskip
	\label{alg:reduce}
\end{algorithm}

Next, we construct a weight vector which is $0$ everywhere except for the first edge of every even walk, and find a point $x$ in the relative interior of $\PM(G, w)$ by applying \cref{lem:avg}. By our choice of weight vector, each even walk either loses an edge or gets blocked by a tight odd set as we will prove in \cref{sec:oddsets}. Our last step consists of finding a number of disjoint odd sets $S_1,\dots, S_l$, such that each even walk $W_i$, that did not lose an edge, has an edge with both endpoints in one $S_j$. We describe the procedure \DisjointOddSets and prove these properties in \cref{sec:oddsets}.

Now we can prove the following.

\begin{lemma}
	\label{lem:reduction}
	After running \Reduce we either find a balanced viable set, or $\card{E}$ gets reduced by a constant factor.
\end{lemma}
\begin{proof}
	We find $k$ even walks where $k=\Omega(\card{E})$. Every walk either loses an edge in the edge removal step, or loses an edge after shrinking $S_1,\dots, S_l$. So the number of edges gets reduced by at least $k$, which is a constant fraction of $\card{E}$ as long as $\card{E}$ is large enough (larger than a large enough constant). Note that we never encounter graphs with $\card{V}<1/c_1$ by \cref{lem:terminate}, so by setting $c_1$ small enough we can assume that $\card{E}$ is larger than a desired constant.
\end{proof}

By \cref{lem:reduction}, our measure of progress, $\card{E}$ gets reduced by a constant factor each time until we find a balanced viable set. Therefore the number of times \Reduce is called is at most $O(\log(\card{E_0}))$, so as long as \DisjointOddSets can be run in \NC, the whole algorithm is in \NC.

We describe the remaining pieces, \Preprocess in \cref{sec:preprocessing}, and \DisjointOddSets in \cref{sec:oddsets}.

\subsection{Preprocessing}
\label{sec:preprocessing}

Here we describe the procedure \Preprocess and prove \cref{lem:preprocess}. The pseudocode is given in \cref{alg:preprocess}. Throughout the process, we make sure that $\card{f^{-1}(v)}<c_1\card{V_0}$ for every $v$ or we find a balanced viable set.

In the first step, we remove any edge of $G$ that does not participate in a perfect matching. Next we make the graph connected. We arrive at a connected graph where every edge participates in a perfect matching. This ensures that there are no vertices of degree $1$, unless the entire graph is a single edge; but in that case we  return $f^{-1}(v)$ as a balanced viable set for any of the two vertices.

After having a connected graph with no vertices of degree $1$, while half of the vertices have degree $2$, we shrink them into other vertices by finding appropriate tight odd sets. The while loop can be run at most a logarithmic number of times, because each time the number of vertices gets reduced by a factor of $2$.

\begin{algorithm}
	\caption{The preprocessing step.}
	
	\Preprocess($G=(V, E)$, $f:V_0\to V$)

	Let $x\leftarrow \avg(\PM(G))$.
	
	\ParallelFor{$e\in E$ with $x_e=0$}{
		Remove $e$ from $G$.
	}
	
	Let $G, f\leftarrow \MakeConnected(G, f)$.
	
	\While{$\card{\set{v\in V\given \deg(v)=2}}>\card{V}/2$}{
		Let $G,f\leftarrow \ShrinkDegreeTwos(G, f)$.
	}
	
	\bigskip
	\label{alg:preprocess}
\end{algorithm}

It remains to describe the procedures \MakeConnected and \ShrinkDegreeTwos. Both of these procedures work by shrinking tight odd sets w.r.t.\ $x$. In both, we have to be slightly careful to avoid shrinking a large piece of the original graph causing a violation of the condition $\card{f^{-1}(v)}<c_1\card{V_0}$.

First let us describe \MakeConnected. We first find the connected components $C_1,\dots,C_k$ of $G$. We sort them to make sure $\card{f^{-1}(C_1)}\leq \dots\leq \card{f^{-1}(C_k)}$. Let $v$ be an arbitrary vertex of $C_k$. For any $i<k$ the set $S_i=\set{v}\cup C_1\cup \dots C_i$ is a tight odd set, because $\set{v}$ is a tight odd set and adding entire connected components does not change the cut value. If $\card{f^{-1}(S_{k-1})}<c_1\card{V_0}$, then we can simply shrink $S_{k-1}$ into a single vertex and make the graph connected. Otherwise let $j$ be the first index where $\card{f^{-1}(S_j)}\geq c_1\card{V_0}$. Then $f^{-1}(S_j)$ is a viable set, because it is a tight odd set. We claim that it is balanced as well; for this we need to show that $\card{f^{-1}(S_j)}\leq (1-c_1)\card{V_0}$. We have
\[ \card{f^{-1}(S_j)}=\card{f^{-1}(S_{j-1})}+\card{f^{-1}(C_j)}\leq c_1\card{V_0}+\frac{1}{2}\card{V_0}, \]
where we used the fact that $C_j$ is not the largest component in terms of $f^{-1}(C_j)$. So as long as $c_1+1/2<1-c_1$, we are done. This is clearly satisfied for small enough $c_1$.

Now let us describe \ShrinkDegreeTwos. First we identify all vertices of degree $2$. Some of these vertices might be connected to each other, in which case we get paths formed by these vertices. We can extend these paths, by the doubling trick in polylog time to find maximal paths consisting of degree $2$ vertices. Then, in parallel, for each such maximal path we do the following: Let the vertices of the path be $(v_1,\dots,v_k)$. Further, let $v_0$ be the vertex we would get if we extended this path from the $v_1$ side and $v_{k+1}$ the one we would get from the $v_k$ side. Note that $\deg(v_i)=2$ for $i=1,\dots,k$ but not for $i=0,k+1$.

We claim that for any even $i$, the set $S_i=\set{v_0,v_1,\dots,v_i}$ is a tight odd set. To see this, let $t=x_{(v_0,v_1)}$. Then because $v_1$ has degree $2$, it must be that $x_{(v_1,v_2)}=1-t$. Then, this means that $x_{(v_2,v_3)}=t$, and so on. In the end, we get that $x_{(v_{i-1},v_i)}=1-t$. Now, look at the edges in $\delta(S)$. They are either adjacent to $v_0$ or $v_i$. Those adjacent to $v_0$ have a total $x$ value of $1-t$ and those adjacent to $v_i$ have a total $x$ value of $t$. So $\ds[S_i]{x}=t+(1-t)=1$.

Now let $j$ be the first even index such that $\card{f^{-1}(S_j)}\geq c_1\card{V_0}$. If no such index exists, we can simply shrink $S_k$ or $S_{k+1}$ (depending on the parity of $k$). Else, we claim that $S_j$ is a balanced viable set. Viability follows from being a tight odd set. Being balanced follows because
\[ \card{f^{-1}(S_j)}=\card{f^{-1}(S_{j-2})}+\card{f^{-1}(v_{j-1})}+\card{f^{-1}(v_{j})}\leq c_1\card{V_0}+c_1\card{V_0}+c_1\card{V_0}. \]
So as long as $3c_1\leq (1-c_1)$, the set $S_j$ is balanced and we can simply return it.

Having all of the ingredients, we now finish the proof of \cref{lem:preprocess}.

\begin{proof}[Proof of \cref{lem:preprocess}]
	It is easy to see that the point $x\in \PM(G)$ remains a valid point throughout, i.e., it remains in the matching polytope even after shrinking sets. This is because we only shrink tight odd sets w.r.t.\ $x$. Assume that the algorithm does not find a balanced viable set.
	
	After \MakeConnected the graph becomes connected, and from then on it remains connected.
	
	Since $x$ remains a valid point in the matching polytope until the end, every edge at the end participates in a perfect matching. But in a connected graph, this means that there are no vertices of degree $1$.
	
	Finally note that by the stopping condition of the while loop, the algorithm terminates only when at most half of the remaining vertices have degree $2$.
\end{proof}
	\section{Tight Odd Sets}
\label{sec:oddsets}

In this section we describe the main remaining piece of the algorithm, namely the procedure \DisjointOddSets. The input to this procedure is a graph $G=(V,E)$ and a map $f:V_0\to V$, a number of edge-disjoint even walks $W_1,\dots,W_m$ in $G$, the point $x= \avg(\PM(G, w))$, where $w$ is the weight vector constructed in \cref{alg:reduce}. Note that $x_e>0$ for all $e\in E$, since we removed all edges $e$ with $x_e=0$. We will prove the following:

\begin{lemma}
	\label{lem:mainoddsets}
	There is an \NC~algorithm \DisjointOddSets, that either finds a balanced viable set, or finds disjoint tight odd sets $S_1,\dots, S_l$ satisfying the following: In any $W_i$ there is an edge $e$ both of whose endpoints belong to some $S_j$. Furthermore $\card{f^{-1}(S_j)}<c_1\card{V_0}$ for all $j$.
\end{lemma}

At a high level the procedure works as follows:
\begin{enumerate}
	\item First, for each even walk $W_i$, we find a tight odd set blocking it.
	\item The resulting tight odd sets might cross each other in arbitrary ways. We uncross them to obtain $S_1,\dots,S_l$, being careful not to produce sets with $\card{f^{-1}(S_i)}\geq c_1\card{V_0}$.
\end{enumerate}

In \cref{sec:find-one} we describe the procedure for finding a tight odd set blocking an even walk. Then in \cref{sec:uncross-all}, we describe how to uncross these and produce disjoint odd sets.

\subsection{Finding a tight odd set blocking an even walk}
\label{sec:find-one}

In this section we describe how to find a tight odd set blocking a given even walk $W$. At a high level, we first move slightly outside of the polytope by moving along a direction defined by $W$. Then we find one of the violated constraints defining the matching polytope. This must be the tight odd set we were after.

Recall that an even walk is either a simple even length cycle in $G$ or the following structure: Let $C_1$ and $C_2$ be two odd length edge-disjoint cycles in $G$ and let $P$ be a path connecting vertex $v_1$ of $C_1$
to vertex $v_2$ of $C_2$; if $v_1 = v_2$, $P$ will be the empty path. Starting from $v_1$, traverse $C_1$, then $P$ from $v_1$ 
to $v_2$, then traverse $C_2$, followed by $P$ from $v_2$ to $v_1$.  

Next we define the alternating vector of an even walk $W$. For this purpose, write $W$ as a list of edges $W=(e_1,\dots,e_k)$, where $k$ is even and if the walk contains a path, then the edges of the path will be repeated twice in this list. We define the alternating vector associated to $W$ as the vector $\chi_W$ given by
\[ \chi_W=-\1_{e_1}+\1_{e_2}-\1_{e_3}+\ldots+\1_{e_k}=\sum_i (-1)^i \1_{e_i}. \]
In terms of its components we have
\[ (\chi_W)_e=\begin{cases}
	(-1)^i & \text{if $e=e_i$ and $e$ is not on the path in $W$},\\
	2(-1)^i & \text{if $e=e_i$ and $e$ is on the path in $W$},\\
	0 & \text{if $e\neq e_i$ for any $i$}.\\
\end{cases}\]

Note that for a weight vector $w$, we have
\[ \dotp{w, \chi_W}=-w_{e_1}+w_{e_2}-w_{e_3}+\ldots+w_{e_k}. \]
In particular for the weight vector chosen in \cref{alg:reduce}, we have $\dotp{w, \chi_W}<0$.

We next define the notion of {\em rotation of an even walk}. For a given reference point $x\in\R^E$ an $\epsilon$-rotation by $W$ is simply the point $y=x+\epsilon\chi_W$. We remark that $\epsilon$ will always have a small, though still inverse exponentially large, magnitude. As a simple observation, note that
\[ \dotp{w, y}=\dotp{w, x}+\epsilon\cdot \dotp{w, \chi_W}<\dotp{w, x}. \]

Note that the point $x$ is $\avg(\PM(G, w))$, i.e., we have
\[ x = {\frac {\1_{M_1}+ \ldots + \1_{M_m}}  m},\]
where $M_1, \ldots, M_m$ are all the minimum weight perfect matchings in $G$.

We will now see what happens to an $\epsilon$-rotation of this point if $\epsilon$ is small enough.

\begin{lemma}
	\label{lem:rotation}
	Let $x=\avg(\PM(G,w))$ for some weight vector $w$. Let $W$ be an even walk whose edges are in the support of $x$, i.e., for every $e\in W$, we have $x_e>0$, and let $\dotp{w, \chi_W}<0$. Let $K(n)\leq n^n$ denote the number of perfect matchings in the complete graph $K_n$, and let $y$ be an $\epsilon$-rotation of $x$ with the walk $W$ for some $\epsilon<1/2nK(n)$. Then, the following hold:
	\begin{enumerate}
		\item \label{cond:1} For every vertex $v$, we have $\dv{y}=1$.
		\item \label{cond:2} For every odd set $S\subset V$, if $\ds{x}>1$, then $\ds{y}\geq 1$.
		\item \label{cond:3} For every edge $e\in E$, we have $y_e\geq 0$.
	\end{enumerate}
\end{lemma}
\begin{proof}
	Condition~\ref{cond:1} holds because $\dv{\chi_W}=0$. This identity holds, because the walk $W$ enters and exits each vertex $v$ the same number of times, and the entries and exits have alternating signs, cancelling each other.
	
	Condition~\ref{cond:2} holds, because when $\ds{x}>1$, then it is larger than $1$ by a margin; choosing $\epsilon$ small enough will not let us erase more than this margin. Formally we have
	\[ \dotp{\1_{\delta(S)}, x}=\frac{\dotp{\1_{\delta(S)}, \1_{M_1}}+\dots+\dotp{\1_{\delta(S)}, \1_{M_m}}}{m}, \]
	and note that $\ds{\1_{M_i}}$ is at least $1$ and must be greater than $1$ for some $i$. For that particular $i$ this value must be at least $2$ (in fact, at least $3$), which gives us
	\[ \ds{x}\geq 1+\frac{1}{m}. \]
	Now, note that $\norm{\chi_W}_1\leq 2n$ and $\norm{\1_{\delta(S)}}_\infty\leq 1$ which together imply that
	\[ \abs{\ds{\chi_W}}\leq 2n. \]
	Finally, piecing things together, we have
	\[ \ds{y}=\ds{x}+\epsilon\ds{\chi_W}\geq 1+\frac{1}{m}-2n\epsilon\geq 1+\frac{1}{m}-\frac{2n}{2nK(n)}\geq 1. \]
	
	Condition 3 holds, because again, $x_e>0$ implies that $x_e$ is positive by a margin. We have
	\[ x_e=\frac{(\1_{M_1})_e+\dots+(\1_{M_m})_e}{m}\geq \frac{1}{m}, \]
	which implies that $y_e\geq 1/m-2\epsilon\geq 0$.
\end{proof}

\Cref{lem:rotation} almost ensures that the point $y$ is inside the matching polytope $\PM(G)$ if the starting point $x$ was in $\PM(G)$. The only way that $y$ cannot be in $\PM(G)$ is if there is an odd set $S\subset V$ such that $\ds{x}=1$, i.e., a tight odd set, whose constraint gets violated by $y$. This leads us to the following important lemma, which enables us to extract a tight odd set {\em blocking} the rotation of the walk $W$.

\begin{lemma}
\label{lem:main}
	Suppose that $w$ is a weight vector, $x=\avg(\PM_w(G))$, $W$ is a walk that satisfies the conditions of \cref{lem:rotation}, and furthermore $\dotp{w, \chi_W}<0$. Then there must be an odd set $S\subset V$ such that $\ds{x}=1$ and $\ds{\chi_W}\neq 0$. Furthermore such an $S$ can be found by first obtaining $y$ as an $\epsilon$-rotation of $x$ by $W$, for a small but inverse exponentially large $\epsilon$, and then finding a minimum odd cut in $y$:
	\[ \argmin_{S\subset V,\card{S}\text{ is odd}} \ds{y}. \]	
\end{lemma}
\begin{proof}
	Since $\dotp{w, \chi_W}< 0$ we have $\dotp{w, y}<\dotp{w, x}$. We choose the magnitude of $\epsilon$ to be small enough that the conditions of \cref{lem:rotation} are satisfied. Now, since $x$ was a minimizer of the linear function $x\mapsto \dotp{w, x}$ over the polytope $\PM(G)$, it must be the case that $y\notin\PM(G)$.
	
	Therefore one of the constraints defining the matching polytope, \cref{eq:polytope}, must not be satisfied for $y$. But \cref{lem:rotation} ensures that almost all of these constraints are satisfied; the only possible constraint being violated would be an odd set $S$ such that $\ds{x}=1$ and $\ds{y}<1$. Take any such set $S$ where $\ds{x}=1$ and $\ds{y}<1$. We have
	\[ \ds{y}=\ds{x}+\epsilon\ds{\chi_W}, \]
	which means that $\ds{\chi_W}\neq 0$. In other words, $S$ satisfies the statement of the lemma.
	
	It only remains to show that if we take $S$ to be a minimum odd cut in $y$, then $S$ satisfies $\ds{x}=1$ and $\ds{y}<1$. We know that the only possible constraint being violated by $y$ is an odd set constraint, so for the minimum odd cut it must be true that $\ds{y}<1$. On the other hand if $\ds{x}>1$, then we would get a contradiction from condition~\ref{cond:2} of \cref{lem:rotation}, because that would imply $\ds{y}\geq 1$. So such a set must satisfy $\ds{x}=1$ and $\ds{y}<1$.
\end{proof}

We will say that an odd set $S$ such that $\ds{x}=1$ and $\ds{\chi_W}\neq 0$, is a set that {\em blocks} the walk $W$. By combining the following lemma with \cref{lem:main}, we get that we can find a tight odd set blocking each of our even walks.

\begin{lemma}
	\label{lem:minoddcut}
	There is an \NC~algorithm that given a weight planar graph $G$, outputs the minimum odd cut of $G$.
\end{lemma}

We will prove \cref{lem:minoddcut} in \cref{sec:GH}.

\subsection{Uncrossing tight odd sets}
\label{sec:uncross-all}

Suppose we are given a list of tight odd sets $S_1,\dots,S_m$ that could cross each other in arbitrary ways.

Note that we can assume from the beginning that for each $i$, $\card{f^{-1}(S_i)}\leq \frac{1}{2}\card{V_0}$. If not, we simply replace $S_i$ by $V-S_i$. We can even further assume that $\card{f^{-1}(S_i)}<c_1\card{V_0}$; otherwise, we would return $f^{-1}(S_i)$ as a balanced viable set and end the procedure. Throughout the algorithm we maintain this property.

Our goal is to {\em uncross} the sets $S_1,\dots, S_m$, so that we can shrink all of them at the same time. We make progress from shrinking these sets by making sure that each of our even walks has an edge inside at least one of the shrunk sets, so that shrinking reduces the number of edges by at least the number of walks.

Unfortunately, having an edge inside an $S_i$ is not a property that is preserved by uncrossing. Instead, we require a stronger property that implies having an edge in one $S_i$, and show that this stronger property is preserved by uncrossing. Throughout this section we assume that $x$ is some fixed point in $\PM(G)$ with $x_e>0$ for all $e\in E$.

\begin{definition}
	\label{def:U}
	For a set $S\subseteq V$, define $\Lambda(S)\subseteq \R^E$ to be the linear subspace defined as the span of cut indicators of all tight odd sets contained in $S$:
	\[ \Lambda(S):=\span\{\1_{\delta(T)}\mid T\subseteq S, \card{T}\text{ is odd},\ds[T]{x}=1\}. \]
	We extend this definition to more than one set $S_1,\dots,S_m$ by letting
	\[ \Lambda(S_1,\dots,S_m):=\Lambda(S_1)+\dots+\Lambda(S_m). \]
\end{definition}
We also use the notation $\Lambda^\perp(S_1,\dots,S_m)$ to denote the subspace of $\R^E$ orthogonal to $\Lambda(S_1,\dots,S_m)$.

Next, we will show that $\chi_W$ not being orthogonal to $\Lambda(S_1,\dots,S_m)$ implies that $W$ has an edge in one $E(S_i)$.
\begin{lemma}
	Let $W$ be an even walk, and assume that $\chi_W\notin \Lambda^\perp(S_1,\dots,S_m)$. Then there is at least one edge $e\in W$ and at least one $i$ such that $e\in E(S_i)$.
\end{lemma}
\begin{proof}
	It is easy to see that $\chi_W\notin \Lambda^\perp(S_1,\dots,S_m)$ implies that there is at least one $i$ such that $\chi_W\notin \Lambda^\perp(S_i)$. It follows from \cref{def:U} that there must be some tight odd set $T\subseteq S_i$ such that $\ds[T]{\chi_W}\neq 0$. We will show that $\ds[T]{\chi_W}\neq 0$ implies that there is some $e\in W$ such that $e\in E(T)\subseteq E(S_i)$.
	
	Suppose the contrary, that no edge $e\in W$ is in $E(T)$. Let $W=(e_1,\dots,e_k)$ and note that
	\[ \ds[T]{\chi_W}=\sum_{j=1}^k (-1)^j \ds[T]{\1_{e_j}}. \]
	Every time that $W$ enters a vertex $v\in T$, it must leave immediately from $T$, or else we would find an edge $e\in E(T)\cap W$. Therefore we can pair up the nonzero $\ds[T]{\1_{e_j}}$s into consecutive pairs, possibly pairing up the last edge with the first. Since these pairs appear in the sum with alternating signs, they cancel each other, giving us
	\[ \ds[T]{\chi_W}=0, \]
	which is a contradiction. Therefore $W$ must have at least one edge in $E(T)\subseteq E(S_i)$.
\end{proof}

Next we will define our basic {\em uncrossing} operations and show that they preserve this nonorthogonality property. Whenever we have two tight odd sets $S_1$ and $S_2$ we will show that we can uncross them, i.e., replace them by new tight odd sets without shrinking the subspace $\Lambda(S_1)+\Lambda(S_2)$. We will use the following uncrossing lemma, which is standard in the literature. We will prove it for the sake of completeness.

\begin{lemma}
	\label{lem:submodularity}
	If $S_1$ and $S_2$ are tight odd sets then either $S_1\cap S_2, S_1\cup S_2$ are tight odd sets and
	\[ \1_{\delta(S_1)}+\1_{\delta(S_2)}=\1_{\delta(S_1\cap S_2)}+\1_{\delta(S_1\cup S_2)}, \]
	or $S_1-S_2$ and $S_2-S_1$ are tight odd sets and
	\[ \1_{\delta(S_1)}+\1_{\delta(S_2)}=\1_{\delta(S_1-S_2)}+\1_{\delta(S_2-S_1)}. \]
\end{lemma}
\begin{proof}
	The following identity holds for any $S_1$ and $S_2$ and can be easily checked by considering all possible configurations of the endpoints of an arbitrary edge:
	\[ \1_{\delta(S_1)}+\1_{\delta(S_2)}=\1_{\delta(S_1\cap S_2)}+\1_{\delta(S_1\cup S_2)}+2\1_{\delta(S_1-S_2, S_2-S_1)}. \]
	We have two cases: Either $\card{S_1\cap S_2}$ is odd, or it is even.
	
	Case 1: Assume that $\card{S_1\cap S_2}$ is odd. It follows that $\card{S_1\cup S_2}$ is also odd. Then by taking the dot product with $x$ we get
	\[ 1+1=\ds[S_1]{x}+\ds[S_2]{x}\geq \ds[S_1\cap S_2]{x}+\ds[S_1\cup S_2]{x}\geq 1+1, \]
	where the last inequality follows from the fact that $x\in \PM(G)$ and that $S_1\cap S_2$ and $S_1\cup S_2$ are odd sets. Since this inequality is tight it must be the case that $\ds[S_1\cap S_2]{x}=\ds[S_1\cup S_2]{x}=1$, which proves that $S_1\cap S_2$ and $S_1\cup S_2$ are tight odd sets. It further follows that
	\[ \ds[S_1-S_2,S_2-S_1]{x}=0, \]
	which implies that $\1_{\delta(S_1-S_2,S_2-S_1)}=0$, i.e., $\delta(S_1-S_2, S_2-S_1)=0$; this is because $x$ has strictly positive entries. Now we have the desired identity
	\[ \1_{\delta(S_1)}+\1_{\delta(S_2)}=\1_{\delta(S_1\cap S_2)}+\1_{\delta(S_1\cup S_2)}. \]
	
	Case 2: Now assume that $\card{S_1\cap S_2}$ is even. We can replace $S_2$ by $V-S_2$, since $V-S_2$ is also a tight odd set. But now $S_1\cap (V-S_2)=S_1-S_2$ which is an odd set. So it follows from the proof of case 1 that $S_1\cap (V-S_2)$ and $S_1\cup (V-S_2)$ are both tight odd sets and we have
	\[ \1_{\delta(S_1)}+\1_{\delta(S_2)}=\1_{\delta(S_1\cap(V-S_2))}+\1_{\delta(S_1\cup(V-S_2))}. \]
	Now observe that $S_1\cap (V-S_2)=S_1-S_2$ and $S_1\cup (V-S_2)=V-(S_2-S_1)$. Since taking complements does not change either $\delta(\cdot)$ or being a tight odd set, the claim follows.
\end{proof}

Now we use \cref{lem:submodularity} to prove the claim that tight odd sets can be uncrossed without shrinking $\Lambda(S_1)+\Lambda(S_2)$.

\begin{lemma}
	\label{lem:uncross}
	Suppose that $S_1,S_2$ are tight odd sets, i.e., $\card{S_1},\card{S_2}$ are odd and $\ds[S_1]{x}=\ds[S_2]{x}=1$. Then exactly one of the following two conditions holds:
	\begin{enumerate}
		\item \label{cond:union} $S_1\cup S_2$ is a tight odd set and
			\[ \Lambda(S_1)+\Lambda(S_2)\subseteq \Lambda(S_1\cup S_2), \]
		\item \label{cond:diff} $S_1$ and $S_2-S_1$ are both tight odd sets and
			\[ \Lambda(S_1)+\Lambda(S_2)\subseteq \Lambda(S_1)+\Lambda(S_2-S_1). \]
	\end{enumerate}
\end{lemma}

\begin{proof}

	Look at the parity of $\card{S_1\cup S_2}$. If $\card{S_1\cup S_2}$ is odd, then we claim that case~\ref{cond:union} happens. Otherwise, we will show that case~\ref{cond:diff} happens.
	
	Case~\ref{cond:union}: $\card{S_1\cup S_2}$ is odd. In this case $\card{S_1\cap S_2}$ is also odd and it follows by \cref{lem:submodularity} that $S_1\cup S_2$ is a tight odd set. It is trivial from \cref{def:U} that $\Lambda(S_1),\Lambda(S_2)\subseteq \Lambda(S_1\cup S_2)$ which immediately yields
	\[ \Lambda(S_1)+\Lambda(S_2)\subseteq \Lambda(S_1\cup S_2). \]
	
	Case~\ref{cond:diff}: $\card{S_1\cup S_2}$ is even. In this case $\card{S_1-S_2}$ and $\card{S_2-S_1}$ are both odd. Again, from \cref{lem:submodularity} it follows that $S_2-S_1$ is a tight odd set. It remains to prove that $\Lambda(S_1)+\Lambda(S_2)\subseteq \Lambda(S_1)+\Lambda(S_2-S_1)$. It is enough to prove that $\Lambda(S_2)\subseteq \Lambda(S_1)+\Lambda(S_2-S_1)$.
	
	It is enough to show that for any tight odd set $T\subseteq S_2$, we have the inclusion $\1_{\delta(T)}\in \Lambda(S_1)+\Lambda(S_2-S_1)$. We again have two cases: Either $\card{T\cap S_1}$ is odd or even.
	
	If $\card{T\cap S_1}$ is even, it follows from \cref{lem:submodularity} that $T-S_1$  and $S_1-T$ are tight odd sets and
	\[ \1_{\delta(T)}=\1_{\delta(T-S_1)}+\1_{\delta(S_1-T)}-\1_{\delta(S_1)}. \]
	We have $\1_{\delta(S_1-T)},\1_{\delta(S_1)}\in \Lambda(S_1)$ and $\1_{\delta(T-S_1)}\in \Lambda(S_2-S_1)$. So $\1_{\delta(T)}\in \Lambda(S_1)+\Lambda(S_2-S_1)$ as desired.
	
	The only case that remains is when $\card{T\cap S_1}$ is odd. In this case we apply \cref{lem:submodularity} to the sets $T$ and $S_2-S_1$, both of which are tight odd sets. Note that $T\cap (S_2-S_1)=T-S_1$ which has even size by assumption. Therefore by \cref{lem:submodularity}, $(S_2-S_1)-T$ and $T-(S_2-S_1)=S_1\cap T$ are also tight odd sets and
	\[ \1_{\delta(T)}=\1_{\delta(S_2-S_1-T)}+\1_{\delta(S_1\cap T)}-\1_{\delta(S_2-S_1)}. \]
	We have $\1_{\delta(S_1\cap T)}\in \Lambda(S_1)$ and $\1_{\delta(S_2-S_1-T)},\1_{\delta(S_2-S_1)}\in \Lambda(S_2-S_1)$ which proves that $\1_{\delta(T)}\in \Lambda(S_1)+\Lambda(S_2-S_1)$ as desired.
	\end{proof}

Given tight odd sets $S_1,\dots,S_m$, repeated applications of \cref{lem:uncross} allow us to uncross them, i.e., replace them by pairwise disjoint tight odd sets $S_1',\dots,S_{m'}'$ such that $\Lambda(S_1,\dots,S_m)\subseteq \Lambda(S_1',\dots,S_{m'}')$. However, naively applying \cref{lem:uncross} would result in a sequential algorithm which is not in \NC. We will next show how we can do the uncrossing in \NC.

We will use a divide-and-conquer approach to uncross a given list of tight odd sets $S_1,\dots,S_m$. The high-level description of our procedure, \Uncross, is given in \cref{alg:uncross}. We roughly divide the given sets into two parts, and recursively uncross each part. Then we call the procedure \MergeUncross in order to merge the resulting sets.

\begin{algorithm}[H]
	\caption{Divide-and-conquer algorithm for uncrossing tight odd sets}
	\Uncross{$S_1,\dots,S_m$}
	
	\eIf{m=1}{
		\Return $S_1$
	}{
		\InParallel{
			$R_1,\dots, R_{p}\leftarrow$ \Uncross{$S_1,\dots,S_{\lceil m/2\rceil}$}
			
			$C_1,\dots, C_{q}\leftarrow$ \Uncross{$S_{\lceil m/2\rceil+1},\dots,S_m$}
		}
		
		\Return{\MergeUncross{$R_1,\dots,R_p,C_1,\dots,C_q$}}
	}
	\bigskip
	\label{alg:uncross}
\end{algorithm}

Next, we will describe the merging procedure \MergeUncross. The procedure \MergeUncross, similarly to \Uncross, accepts a list of tight odd sets and returns a list of pairwise disjoint tight odd sets whose $\Lambda$ is not smaller. With some abuse of notation, we still name the inputs to \MergeUncross as $S_1,\dots,S_m$. The difference between \MergeUncross and \Uncross is that the input sets to \MergeUncross satisfy certain properties highlighted below.
\begin{lemma}
	\label{lem:3way}
	Suppose that $\{S_1,\dots,S_m\}=\{R_1,\dots,R_p,C_1,\dots,C_q\}$, where $m=p+q$ and $R_1,\dots,R_p$ are pairwise disjoint tight odd sets and $C_1,\dots,C_q$ are also pairwise disjoint tight odd sets. Then $S_1,\dots,S_m$ have no $3$-wise intersections. Furthermore, the intersection graph of $S_1,\dots,S_m$, where two $S_i$'s are connected if they have a nonempty intersection, is bipartite.
\end{lemma}
\begin{proof}
	If we select any three sets $S_i,S_j,S_k$, then either two of them are from $R_1,\dots,R_p$ or two of them are from $C_1,\dots,C_q$. In either case, those two sets would not have any intersections.
	
	It is also easy to see that the intersection graph is bipartite, since $R_1,\dots,R_p$ naturally form one part and $C_1,\dots,C_q$ the other; by assumption, no two sets from the same part have any intersection.
\end{proof}

Having no $3$-way intersections means that we can compute the parity of any union of $S_1,\dots,S_m$ from their pairwise intersections. This is more handily captured by the notion of an intersection parity graph.

\begin{definition}
	For tight odd sets $S_1,\dots,S_m$ satisfying the conditions of \cref{lem:3way}, define the intersection parity graph $H=(V_H,E_H)$, as follows: Let $V_H$, the nodes of $H$, be $S_1,\dots,S_m$ and for $i\neq j$ let there be an edge between $S_i$ and $S_j$ if and only if $\card{S_i\cap S_j}$ is odd.
\end{definition}

An immediate corollary of \cref{lem:3way} is that $H$ is bipartite. Another corollary is that the parity of $\card{\cup_i S_i}$ is the same as the parity of $\card{V_H}+\card{E_H}$ which we simply denote by $\card{H}$; this is because the inclusion-exclusion formula stops at pairwise intersections for our sets. We use the notation $H(S_{i_1},\dots,S_{i_k})$ to denote the induced subgraph on nodes $S_{i_1},\dots,S_{i_k}$. With this notation we have
\[ \card{S_{i_1}\cup \dots S_{i_k}}\eqtwo\card{H(S_{i_1},\dots,S_{i_k})},\]
where $\eqtwo$ represents having the same parity.

By \cref{lem:uncross}, if $S_1,S_2$ have an edge between them in $H$, then the union $S_1\cup S_2$ will also be a tight odd set. If there is a third set $S_3$ connected to $S_2$, we can again include $S_3$ in this union, i.e., $S_1\cup S_2\cup S_3$ will be a tight odd set.

Can we repeatedly apply this procedure and otain $S_1\cup \dots \cup S_m$ as a tight odd set? There seem to be two barriers to this. If the graph $H$ is not connected, we can never take the union of two sets from different connected components. Another natural barrier is that $\card{S_1\cup \dots \cup S_m}$ could possibly be even; so it will never emerge out of this process, because \cref{lem:uncross} only produces {\em odd} tight sets. For simplicity of notation we use $\cup H$ to denote $S_1\cup\dots\cup S_m$.

Surprisingly, the two mentioned barrier are really the only barriers, as we will show next.
\begin{lemma}
	\label{lem:odd}
	Assume that $H=H(S_1,\dots,S_m)$ is connected and that $\card{H}\eqtwo 1$. Then $\cup H=S_1\cup \dots \cup S_m$ is a tight odd set, and $\Lambda(S_1,\dots,S_m)\subseteq \Lambda(\cup H)$.
\end{lemma}
\begin{proof}
	We just need to show that $\cup H$ is a tight odd set. The fact that $\Lambda(S_1,\dots,S_m)\subseteq \Lambda(\cup H)$ is trivial from \cref{def:U}.
	
	We will use induction on $\card{V_H}$ to prove this fact. It is trivial to check this for $\card{V_H}\leq 2$. Even if $\card{V_H}=3$, the only graph that is connected and bipartite on $3$ nodes would be the path of length $2$ and we have already described that in this case we can take the union by two applications of case~\ref{cond:union} from \cref{lem:uncross}.
	
	Now consider a depth-first-search (DFS) tree started from an arbitrary node of $H$. If $S$ is any leaf of this tree with $\deg_H(S)\eqtwo 1$, then we can proceed as follows: The graph $H-\{S\}$ will have one fewer node and odd many fewer edges. Therefore $\card{H-\{S\}}\eqtwo 1$, and obviously $H-\{S\}$ is connected, since $S$ was a leaf. By induction, $\cup(H-\{S\})$ is a tight odd set. But $S$ is also an tight set, and by assumption the union of the two, $\cup(H-\{S\})\cup S=\cup H$, is also odd. So by \cref{lem:uncross} we get that $\cup H$ is a tight odd set. So from now on, assume that for any leaf node $S$, $\deg_H(S)\eqtwo 0$. More generally, if $S$ is any node whose removal does not disconnect the graph, we can assume that $\deg_H(S)\eqtwo 0$, or else we can proceed as before. Note that this implies that any leaf in the tree has at least one back edge, i.e., an edge going to an ancestor other than its parent. This is true, because any leaf must have at least one edge other than the one going to its parent, and in a DFS tree there are no cross edges, which means that this edge must be a back edge.
	
	Note that in a DFS tree, the leaf nodes are never connected to each other. This implies, by simple parity counting, that if $S_1,S_2$ are two leaves then $\card{H-\{S_1,S_2\}}\eqtwo 1$. Note that $H-\{S_1,S_2\}$ is also connected, so by induction $\cup(H-\{S_1,S_2\})$ is a tight odd set.
	
	Now, if the DFS tree has at least four leaves $S_1,S_2,S_3,S_4$, we can proceed as follows: Consider the graphs $H-\{S_1,S_2\}$ and $H-\{S_3, S_4\}$. They both satisfy the assumptions of the induction and therefore $\cup(H-\{S_1,S_2\})$ and $\cup(H-\{S_3,S_4\})$ are both tight odd sets. Their union is again $\cup H$ which has an odd parity. So again by \cref{lem:uncross} we get that $\cup H$ is a tight odd set. From now on we assume that there are at most $3$ leaves in the tree.
	
	If there are any two leaves $S_1,S_2$ that share a parent $P$, we can proceed as follows: The graph $H-\{S_1,S_2\}$ again satisfies the assumptions of induction. We also have that $S_1\cup P\cup S_2$ is a tight odd set; this follows by applying the base case to the subgraph $H(S_1,S_2,P)$ which is a path of length $2$. Again we have two tight odd sets $\cup(H-\{S_1,S_2\})$ and $S_1\cup P\cup S_2$ whose union $\cup H$ is odd. Therefore $\cup H$ is a tight odd set. So from now on, we assume that no two leaves share a parent. 
	
	Now assume that the DFS tree has three leaves $S_1,S_2,S_3$. Without loss of generality, assume that $S_1$ is the deepest leaf. Let $P$ be the parent of $S_1$. Note that $P$ does not have any other children in the tree, because $S_1$ was the deepest leaf and no two leaves share a parent. Note that the removal of $P$ does not disconnect the graph because $S_1$ has a back edge. Therefore it must be that $\deg_H(P)\eqtwo 0$. Note also that $P$ is not connected to $S_2$ or $S_3$, because a DFS tree does not have cross edges. All of this implies that $H-\{S_3,P\}$ is connected, and also has odd parity. As before $H-\{S_1,S_2\}$ also satisfies the assumptions of the induction. So again, we get two tight odd sets whose union is $\cup H$ and therefore $\cup H$ is a tight odd set.
	
	Now assume that the DFS tree has only two leaves $S_1,S_2$. Let $P_1$ be the parent of $S_1$ and $P_2$ the parent of $S_2$. Let $Q$ be the lowest common ancestor of $S_1$ and $S_2$ in the tree. If $P_1,P_2\neq Q$, then we can proceed similarly to the previous case: Both $P_1$ and $P_2$ must have an even degree, since their removal does not disconnect the graph. Now $H-\{P_1,S_2\}$ and $H-\{P_2,S_1\}$ are both connected and have an odd parity. We use induction and the fact that their union is $\cup H$ to again show that $\cup H$ is a tight odd set. So assume that one of $P_1,P_2$ is the same as $Q$. Without loss of generality, assume that $P_2=Q$. Note that $P_1\neq Q$, or else we would have two leaves sharing a parent, which is already a resolved case. Now let $R$ be the parent of $P_2=Q$. Note that $S_2$ has a back edge, but its back edge cannot be to $R$ because that would create a triangle between $S_2,P_2,R$ which is forbidden in our bipartite graph. So the back edge must be to some ancestor of $R$. This means that removing $R$ or even removing both $R, S_1$ does not disconnect the graph. Since removing $R$ does not disconnect the graph we have $\deg_H(R)\eqtwo 0$. Now we have two cases:
	\begin{enumerate}
		\item If $S_1$ does not have an edge to $R$, that would imply $H-\{S_1,R\}$ is odd and connected. Similar to the case of three leaves, we would get that $H-\{P_1,S_2\}$ is odd and connected as well. But then $H-\{S_1,R\}$ and $H-\{P_1,S_2\}$ are two connected and odd subgraphs whose union is $H$ which implies that $\cup H$ is a tight odd set. 
		\item Now assume that $S_1$ does have an edge to $R$. Note that $Q=P_2$ is a parent of $S_2$ and an ancestor of $S_1$. So it must have some other child, which we will call $C$. Note that $C\neq S_1$, or else $S_1,S_2$ would be two leaves sharing a parent, which has already been resolved. Now, the removal of $C$ does not disconnect the graph because of the edge between $S_1$ and $R$. So it must be that $\deg_H(C)\eqtwo 0$. On the other hand, the removal of both $S_2,C$ also does not disconnect the graph. Also note that there is no edge between $S_2$ and $C$ because such an edge would create a triangle $S_2,C,Q$ which is forbidden in our bipartite graph. All of these mean that $H-\{S_2,C\}$ is odd and connected and by induction $\cup(H-\{S_2, C\})$ is a tight odd set. On the other hand $S_2\cup Q\cup C$ is also a tight odd set because the induced graph on these three sets is a path of length $2$. Again we have found two tight odd sets $\cup(H-\{S_2,C\})$ and $S_2\cup Q\cup C$ whose union gives us $\cup H$ and we are done.
	\end{enumerate}
	
	The only remaining case is when the DFS tree has only one leaf, i.e., when the DFS tree is a Hamiltonian path. If the root and the leaf are not connected to each other, we can find another DFS tree such that it has more than one leaf and reduce the problem to the previous cases considered. Consider starting the DFS from the child of the current root and going down the Hamiltonian path until we reach the current child. Since this child was not connected to the root, the DFS procedure cannot continue and has to back up. Eventually the original root will be connected somewhere along the tree as a leaf, but we now have two leaves, and we have already considered this case.
	
	So the only case that remains is if the DFS tree is a Hamiltonian path and that the root is connected to the leaf. This tree with the extra edge gives us a Hamiltonian cycle. Since the removal of any node in this graph does not disconnect the graph, all of the degrees must be even. Note that the entire graph cannot be simply this Hamiltonian cycle, because otherwise $\card{H}\eqtwo m+m\eqtwo 0$. So there must be some edge, other than those of the cycle, between two vertices $P$ and $Q$. Let the two neighbors of $P$ on the Hamiltonian cycle be $A,B$. Note that removing both $A,B$ does not disconnect the graph. There is also no edge between $A$ and $B$, because otherwise we would have a triangle $A,B,P$ which is forbidden in bipartite graphs. So $H-\{A,B\}$ is odd and connected and by induction $\cup(H-\{A, B\})$ is a tight odd cut. Note that $A\cup B\cup P$ is also a tight odd cut, because the induced graph on $A,B,P$ is a path of length $2$. Again we have written $H$ as the union of two connected and odd subgraphs; this implies that $\cup H$ is a tight odd set.
\end{proof}

\Cref{lem:odd} is the powerful pillar we use to create the method \MergeUncross. If the intersection parity graph $H$ has multiple connected components, we can deal with each one separately and then uncross the results using case~\ref{cond:diff} of \cref{lem:uncross}. If all of the connected components have odd parity, then we can take the union in each one and proceed. The only case we still need to show how to handle is when a connected component of $H$ has even parity. We will show next that the even parity case can also be handled very easily.

\begin{lemma}
	\label{lem:even}
	Assume that $H=H(S_1,\dots,S_m)$ is connected and $\card{H}\eqtwo 0$. Then there are two induced subgraphs of $H$, which are both odd and connected, and which together cover every node. Furthermore, these two subgraphs can be found in \NC.
\end{lemma}
\begin{proof}
	We will be working with the biconnected components of $H$ and the corresponding block-cut tree. A biconnected component is simply a maximal subgraph such that the removal of any vertex from it does not disconnect the subgraph. The block-cut tree is formed by introducing a node for each biconnected component and a node for every cut vertex, a vertex whose removal disconnects the graph, and connecting a cut vertex to all biconnected components to which it belongs. Finding biconnected components and forming the block-cut tree can be easily done in \NC. For example in parallel for every pair of edges, and every vertex, one can check whether the removal of that vertex disconnects the pair of edges; then one can form equivalence classes out of the edges and obtain the biconnected components. For more efficient and elegant algorithms in \NC, see \cite{tarjan1985efficient}.
	
	For an induced subgraph $B=(V_B, E_B)$ let us define its inverse parity as the parity of $\card{V_B}+\card{E_B}+1$ and denote this by $\bcard{B}$. Note that we have $\bcard{B}\eqtwo 1+\card{B}$. We regard biconnected components as induced subgraphs, unless otherwise stated. Inverse parity has a certain additivity property. Namely, if $B_1$ and $B_2$ are induced subgraphs that share only a single vertex and have no edges to each other, then $\bcard{B_1\cup B_2}=\bcard{B_1}+\bcard{B_2}$.
	
	Using this, one can easily compute the inverse parity of any subtree of the block-cut tree. In the block-cut tree, to each biconnected component assign its inverse parity, and to each cut vertex assign $0$. Then it is easy to see by the additivity property that for any subtree of the block-cut tree, the inverse parity of the union of all blocks in the subtree is simply the parity of the sum of assigned numbers.
	
	In particular, since $\card{H}\eqtwo 0$, or in other words, $\bcard{H}\eqtwo 1$, there must be an odd number of $1$s in the block-cut tree.
	
	We will first solve the problem when there are at least three $1$s in the tree. In this case, we can find two subtrees whose union is the entire tree, each having an even number of $1$s. This suffices, because the union of all biconnected components in each subtree would be an odd connected graph, and by \cref{lem:odd} we can merge all of the nodes in it. Each subtree will be obtained by simply partitioning the block-cut tree by removing an edge and looking at one of the resulting two subtrees. Clearly we can try all such partitions in \NC. So it remains to show that at least two of them, whose union is the entire tree, have an even internal sum. For this, look at the $1$ nodes in the tree whose distance, in the tree, is the largest. Let them be $B_1$ and $B_2$. Look at the path on the block-cut tree connecting $B_1$ to $B_2$ and let the edge adjacent to $B_1$ be $e_1$ and the one adjacent to $B_2$ be $e_2$. Now if we partition the block-cut tree by removing $e_1$, we get two parts, one of which contains $B_2$, and the other part can only contain one $1$ node, namely $B_1$. Otherwise, the distance between $B_1$ and $B_2$ would not have been maximal. So the subtree containing $B_2$ has an even sum. Similarly if we remove $e_2$ from the block-cut tree, the part containing $B_1$ will have an even sum. It is not hard to see that these two subtrees cover the whole tree.
	
	So the only remaining case is when the block-cut tree has only one $1$ node. In that case let $B$ be the biconnected component with $\bcard{B}\eqtwo 1$.
	
	First consider the case where $B$ is the entire graph $H$. In this case, we will show that either there is a vertex $S$ where $B$ and $B-\{S\}$ are both odd and connected, or there are two vertices $S_1,S_2$ connected by an edge such that $B-\{S_1,S_2\}$ and $H(S_1, S_2)$ are both odd and connected. First, note that if any node in $B$ has an even degree, then this condition is automatically satisfied. Because if $S_1$ is such a node, $B-\{S_1\}$ is connected since $B$ is biconnected. It is also odd because $B-\{S_1\}$ has one fewer node and an even number of fewer edges. So assume from now on that the degree of every node in $B$ is odd. Now we want to obtain the nodes $S_1,S_2$ as described before. This is easy to derive from an open ear decomposition of $B$. Note that $\bcard{B}\eqtwo 1$ implies that $B$ cannot be simply a single edge, so it must have an open ear decomposition. Look at this ear decomposition, and add the ears one by one. Look at the last ear added that was not a single edge. Suppose that this ear was some path $(S_1,\dots,S_k)$. Then, note that $S_2$ is a new node added by this ear, and since no new nodes are added after this ear, the removal of $S_1,S_2$ leaves $B$ connected; even if this ear was the initial cycle, this is still true. So $B-\{S_1,S_2\}$ is connected and since the degrees of $S_1,S_2$ are both odd and they are connected to each other, it must be that $B-\{S_1,S_2\}$ is odd. Since $S_1,S_2$ are connected to each other as well $H(S_1,S_2)$ is also connected and odd as desired. Note that the vertex $S_1$ or the pair of vertices $S_1,S_2$ can be found in \NC~by simply checking all possibilities in parallel.
	
	Now consider the case where $B$ is not the entire graph $H$. In this case we proceed as before, and by looking at the induced subgraph $B$, we find either a node $S_1$ or two connected nodes $S_1,S_2$ such that $B-\{S_1\}$ or $B-\{S_1,S_2\}$ is connected and odd. So we have a partition of $B$ into a single or a pair of vertices and the rest of $B$. We simply attach the biconnected components other than $B$ to one of the partitions, based on the block-cut tree. This ensures that connectivity is preserved, and further, the parity of the partitions is not changed because every biconnected component other than $B$ has inverse parity $0$. Again this operation can be done in \NC, since the partition inside $B$ can be found in \NC, and connecting the rest of the biconnected components is simply a matter of partitioning the block-cut tree into two or three parts.
\end{proof}

Now, armed with \cref{lem:odd,lem:even}, we can describe the procedure \MergeUncross. We will first make sure that even intersections are completely removed, i.e., made empty. This is easy to do in parallel, because there are no $3$-wise intersections. Then we apply \cref{lem:odd} or \cref{lem:even} to each connected component of $H$. To avoid creating sets $S$ with $\card{f^{-1}(S)}\geq c_1 \card{V_0}$, we always pass our new sets through the procedure \CheckBalancedViable, which will potentially find a balanced viable set and end the procedure.

\begin{algorithm}[H]
	\caption{Algorithm for uncrossing partially uncrossed sets}
	\MergeUncross{$S_1,\dots,S_m$}
	
	\ParallelFor{$i=1\dots m$}{
		$S_i\leftarrow S_i-\bigcup_{j<i, \card{S_i\cap S_j}\text{ even}}S_j$
	}
	$H\leftarrow H(S_1,\dots, S_m)$
	
	$H_1,\dots,H_k\leftarrow$\ConnectedComponents{H}
	
	$\cF\leftarrow\emptyset$
	
	\ParallelFor{$i=1\dots k$}{
		\eIf{$\card{H_i}=\card{V_{H_i}}+\card{E_{H_i}}$ is odd}{
			\CheckBalancedViable($H_i$).
			
			Add $\cup H_i$ to $\cF$.
		}{
			Let $H_i', H_i''$ be the two induced subgraphs promised by \cref{lem:even}.
			
			\CheckBalancedViable($H_i'$).
			
			\CheckBalancedViable($H_i''$).
			
			Add $\cup H_i'$ to $\cF$.
			
			Add $\cup H_i''-\cup H_i'$ to $\cF$. 
		}
	}
	\Return $\cF$
	\bigskip
	\label{alg:mergeuncross}
\end{algorithm}

We just have to describe \CheckBalancedViable. The input to this procedure is an odd connected subset $H$ of the intersection parity graph. If $\card{f^{-1}(\cup H)}<c_1\card{V_0}$, then this procedure simply does nothing. Otherwise it outputs a balanced viable set as follows:

If $\card{f^{-1}(\cup H)}\leq (1-c_1)\card{V_0}$, then it simply outputs $f^{-1}(\cup H)$ as the balanced viable set. Otherwise, we order the vertices of $H$ as $S_1',\dots,S_k'$ so that for any $i$, the induced subgraph on $U_i=\set{S_1',\dots,S_i'}$ is connected. For example, sorting according to shortest distance (in $H$) to an arbitrary initial vertex $S_1'$ would satisfy this property. Now let $j$ be the first index for which $\card{f^{-1}(\cup U_j)}\geq 2c_1\card{V_0}$. Then $\card{f^{-1}(\cup U_j)}\leq 3c_1\card{V_0}<(1-c_1)\card{V_0}$. So if $\card{\cup{U_j}}\eqtwo 1$, then we can return $f^{-1}(\cup U_j)$ as a balanced viable set (it is a tight odd set by \cref{lem:odd}). Otherwise by \cref{lem:even}, we can find two subsets of $U_j$ whose union covers $U_j$ and are odd. We simply return the subset with the larger value of $\card{f^{-1}(\cdot)}$ as the balanced viable set.

All together we get the following result:
\begin{theorem}
	\label{thm:uncrossing}
	Given tight odd sets $S_1,\dots,S_m$, there is an \NC~algorithm that either finds a viable set or outputs pairwise disjoint tight odd sets $S_1',\dots,S_{m'}'$ such that
	\[ \Lambda(S_1,\dots,S_m)\subseteq \Lambda(S_1',\dots,S_{m'}'), \]
	and $\card{f^{-1}(S_i')}<c_1\card{V_0}$.
\end{theorem}

	\section{Other Algorithmic Ingredients}
\label{sec:ingredients}

In this section we describe the remaining algorithmic ingredients we used in \cref{sec:main-alg,sec:oddsets}.

\subsection{Finding a point in the relative interior of a face of the matching polytope}
\label{sec:avg}

In this section, we prove \cref{lem:avg} by giving an \NC~algorithm for the following problem: Given a planar graph $G = (V, E)$ and a weight vector on edges $w\in \Z^E$ given in unary, 
find a point, $x$, in the interior of 
    $\PM(G,w)$, where $\PM(G,w)$ denotes the face of the perfect matching polytope of $G$ containing all minimum weight
    perfect matchings in $G$ and their convex combinations (clearly, the corner points of this face are precisely the set of minimum weight
    perfect matchings in $G$).
\begin{proof}[Proof of \cref{lem:avg}]
    Let $\#G_w$ denote the number of minimum weight perfect matchings in $G$ w.r.t. edge weights $w$, and for each 
edge $e \in E$, let $\#G^e_w$ denote the number of such matchings which contain the edge $e$. 
The point $x$ we will find will have coordinate 
\[  x_e = \frac{\#G^e_w}{\#G_w} .\]
Clearly, $x$ satisfies all required conditions. Additionally, observe that if $M_1, \ldots, M_m$ are
all the minimum weight perfect matchings in $G$, then 
\[ x = {\frac {\1_{M_1} + \ldots + \1_{M_m}}  m}=\avg(\PM(G, w)).\]
    
    We will crucially use the fact that a Pfaffian orientation of $G$ can be computed in \NC. 
Let $(i, j) \in E$, with $i < j$. If in the 
Pfaffian orientation, this edge is directed from $i$ to $j$, then let $B_{ij} = y^{w_e}$, otherwise let $B_{ij} = -y^{w_e}$,
    where $y$ is an indeterminate. Let $B$ be the resulting matrix.
Observe that the exponents of the entries of $B$ are polynomially bounded in the input size and hence its determinant can be computed in \NC~\cite{BCP}.
    Consider the lowest degree term in $\det(B)$; let its degree be $d$. Then the coefficient of $y^d$ is the square of the number of perfect matchings of minimum weight in $G$, i.e.,
    it is $(\#G_w)^2$.
    
    Next, for each edge $e \in E$, we will compute $\#G^e_w$, the number of minimum weight perfect matchings that edge $e$ participates in. Zero out the 
    two entries in $B$ corresponding to $e$ to obtain matrix $B_e$ and compute $\det(B_e)$. Then the coefficient of $y^d$ will 
be $(\#G_w - \#G^e_w)^2$.
    Hence, $\#G^e_w$ as well as $\#G_w$ can be computed.
    Clearly, this can be done in parallel for all edges.
\end{proof}

\subsection{Finding linearly many edge-disjoint even walks}
\label{sec:evenwalks}

In this section we prove \cref{lem:evenwalks} by showing how to find $\Omega(\card{E})$ many edge-disjoint even walks in a given graph $G=(V,E)$ in \NC. By assumption, $G$ is a connected planar graph that does not have any vertices of degree $1$ and at most $\card{V}/2$ vertices of degree $2$. We first find linearly many edge-disjoint planar faces in $G$.

\begin{lemma}[Adapted from \cite{MahajanV}]
	\label{lem:edge-disjoint-faces}
	There is an \NC~algorithm that returns $\card{E}/288$ edge-disjoint planar faces of a graph satisfying the assumptions of \cref{lem:evenwalks}.
\end{lemma}

\begin{proof}
It is easy to see that the graph has $\Omega(\card{E})$ faces. By Euler's formula we have
\[ \card{V}-\card{E}+\card{F}=2, \]
where $F$ denotes the set of (planar) faces. By rearranging and using the fact that $\deg(v)\geq 2$ for every $v$, we get
\begin{multline*}
	\card{F}\geq \card{E}-\card{V}=\sum_{v\in V}\frac{\deg(v)-2}{2}\geq \sum_{v\in V:\deg(v)\geq 3}\frac{\deg(v)}{6}=\frac{1}{3}(\card{E}-\card{\set{v:\deg(v)=2}})\\
	\geq \frac{1}{3}(\card{E}-\card{V}/2)\geq \frac{1}{6}\card{E}.
\end{multline*}

Consider the planar dual $G^*$ of $G$. Corresponding to each face in $G$, the dual has a vertex, and corresponding to each edge in the primal, there is an edge in the dual. The sum of the degrees in the dual graph is $2\card{E}\leq 12\card{F}$. In other words, the average degree in the dual graph is at most $12$. By Markov's inequality, at least half of the dual vertices must have degree at most $24$. We simply drop the dual vertices of degree more than $24$ from the dual graph and find a maximal independent set in the remaining dual. This can be done in \NC, by \cref{lem:maximal}. The remaining dual has maximum degree at most $24$, so its maximal independent set has at least at least $1/24$ of its vertices, i.e., at least $\card{F}/48\geq \card{E}/288$.
\end{proof}

If at least half of the faces found by \cref{lem:edge-disjoint-faces} are even, we work with these as our even walks. Else, we need to pair up odd faces together with an edge-disjoint path connecting each pair to get $\Omega(\card{E})$ even walks of the second type.

\begin{lemma}
	\label{lem:disjoint-walks}
	Given an even number $f$ of edge-disjoint odd faces in a planar graph, we can find, in \NC, $f^2/16\card{E}$ edge-disjoint even walks, each formed by joining two of the given faces.
\end{lemma}

\begin{proof}
	First we find a spanning tree $T$ of $G$. We will only use paths on the spanning tree to pair up odd faces. For each given odd face, place a token at one of its vertices, arbitrarily. Now we have $f$ tokens on the spanning tree $T$. In \cref{lem:tree-pairing} we will prove that these tokens can be paired up by edge-disjoint paths from the tree in \NC.
	
	We use this pairing of tokens and the paths from the tree $T$ to pair the given odd faces. When we connect two odd faces $O_1$ and $O_2$ by a path $P$, the path $P$ might intersect or even use the edges of $O_1$ and $O_2$. We fix this by replacing $P$ with a subpath of $P$. More precisely, we find the last intersection of $P$ with $O_1$, and the first intersection after that point with $O_2$, and replace $P$ by the subpath between these two intersections. Now the path is edge-disjoint from the $O_1,O_2$.
	
	So far, we have created $f/2$ even walks; but they are not necessarily edge-disjoint. The only way that two of these walks can intersect each other is if the connecting path from one shares an edge with an odd face of the other. Because of this, any given edge $e$ can appear in at most $2$ of the walks. So the average number of edges in an even walk is at most $2\card{E}/f$. Markov's inequality implies that at least $f/4$ of the even walks have at most $4\card{E}/f$ edges. Any of these even walks shares an edge with at most $4\card{E}/f$ other walks, since an edge appears in at most two walks.
	
	By \cref{lem:maximal}, we can find, in \NC, a maximal independent set of these short even walks (an independent set is a just a set of walks that are pairwise edge-disjoint). The number of walks in this independent set will be at least
	\[ \frac{f}{4}\cdot \frac{f}{4\card{E}}=\frac{f^2}{16\card{E}}. \]
	
\end{proof}

We now describe the missing part from the above proof.

\begin{lemma}
	\label{lem:tree-pairing}
	Consider a tree $T$ and an even number of tokens $o_1,\dots,o_f$ placed on the vertices of the tree. We can find, in \NC, a pairing of the tokens using the shortest path on the tree, so that no two paths share an edge.
\end{lemma}
\begin{proof}
	For each edge $e\in T$, we will count the number of tokens on either side of $T$ when $e$ is removed. Since there are an even number of tokens, this count must either be odd on both sides or even on both sides. We will do this in parallel for every edge. We then remove all of the edges whose token counts were even-even.
	
	After this operation, the degree of every vertex $v$ must have the same parity as the number of tokens on it. This is because for each edge $e$ adjacent to $v$, the number of tokens on the other side of $e$ is odd. So the number of tokens not on $v$ has the same parity as $\deg(v)$. But the total number of tokens on the entire tree is even, so this parity is also shared by the number of tokens on $v$.
	
	Now we do the following in parallel for each vertex $v$: We pair up all the tokens on $v$ in an arbitrary way until there is at most one token left. We will then pair the remaining token, if any, with one of the remaining edges; there must be at least one edge if there is at least one token. Now there are an even number of edges adjacent to $v$ that remain. We pair them up in an arbitrary way, so that whenever we use an edge in a pair to enter $v$ we exit using the other edge.
	
	Now by following the paths from each token to the edge it is assigned to, we will get to another token, and this gives us a pairing between tokens. Note that this path following does not have to be done sequentially, but can instead be done using the doubling trick to get an \NC~algorithm.
\end{proof}

We now have the ingredients needed to finish the proof of \cref{lem:evenwalks}.

\begin{proof}[Proof of \cref{lem:evenwalks}]
	We first find $\card{E}/288$ edge-disjoint faces by invoking \cref{lem:edge-disjoint-faces}. If at least half of these faces are even, we return this half. Otherwise we invoke \cref{lem:disjoint-walks}. We have $\card{E}/576$ odd faces, any by possibly dropping one of them we can supply an even number $f\geq \card{E}/576-1$ of odd faces to the algorithm described by \cref{lem:disjoint-walks} and obtain $(\card{E}/576-1)^2/16\card{E}=\Omega(\card{E})$ edge-disjoint even walks. This finishes the proof.
\end{proof}

\subsection{Finding Gomory-Hu trees and minimum odd cuts}
\label{sec:GH}
    
    In this section, we will give an \NC~algorithm for constructing a Gomory-Hu tree for a planar graph $G = (V, E)$ with
    edge weights given by $w : E \rightarrow \R_{\geq 0}$ and finding a minimum odd cut. We will crucially use the fact that an $s$-$t$ max-flow 
and min-cut can be computed
    in a planar graph in \NC~\cite{johnson}. For each pair of vertices $u, v \in V$, let $f(u, v)$ denote the weight of a minimum $u$-$v$
cut in $G$.

Note that if $(S,\overline{S})$ is a minimum $u-v$ cut, then $S$ must consist of a number of connected components of $G$ together with an internally connected subset of vertices. This is because if $S$ contains two disjoint sets $S_1,S_2$ that have no edges to each other, we can find a smaller $u-v$ cut by either taking $S-S_1$ or $S-S_2$ depending on which one still contains $u$. In any case, the graph obtained by shrinking $S$ in $G$ will always remain planar.
    
The sequential algorithm for constructing a Gomory-Hu tree has, at any point, a tree $T$ defined on a partition $S_1, \ldots, S_k$
of $V$, and a weight function $w'$ defined on the edges of $T$. The starting partition is simply $V$, with $T$ having no edges.
The partition and $T$ satisfy: 
\begin{itemize}
\item
For each edge $(S_i, S_j) \in T, \ \exists \ u \in S_i, \ v \in S_j$ such that $w'(S_i, S_j) = f(u, v)$.

\item
The removal of edge $(S_i, S_j)$ from $T$ disconnects $T$. This splits the partitions into two sets, and naturally defines a cut,
say $(S, \overline{S})$ in $G$. This cut must be a minimum $u$-$v$ cut in $G$.
\end{itemize}

In each iteration, the sequential algorithm refines the tree by {\em splitting} one of the partitions into two as follows. It
picks a partition having at least two vertices, say $S_i$. Let $u, v \in S_i$. Let $T_1, \ldots T_l$ be the subtrees of $T$ incident
at node $S_i$. By {\em shrinking} subtree $T_j$ we mean identifying all vertices in $T_j$ and replacing it by single vertex $t_j$.
All edges incident at vertices in $T_j$ from outside $T_j$ are now incident at $t_j$, with the same weight as before. 
Shrinking $T_1, \ldots T_l$ gives a graph on $S_i \cup \{t_1, \ldots, t_l \}$. Let this graph be $G'$; clearly it will
be planar. In $G'$, find a minimum $u$-$v$ cut.
It is easy to show that the weight of this cut will also be $f(u, v)$.

This cut will partition $S_i$ into two sets, say $S'$ and $S''$, with $u \in S'$ and $v \in S''$.
Replace $S_i$ by these two sets to obtain a partition on $k+1$ sets.
The new tree will contain the edge $(S', S'')$ with weight $w'(S', S'') = f(u, v)$. Next, among the subtrees $T_1, \ldots T_l$ take the ones 
on the $u$ side ($v$ side) of the cut and let them be incident at $S'$ ($S''$). The algorithm ends when each partition is a singleton
vertex. The tree so found will be a Gomory-Hu tree.

We now give our \NC~algorithm. The main difference lies in the way set $S_i$ is split. We first define the notion
of a central vertex for $S_i$. Pick a vertex $r \in S_i$ and for each remaining
vertex $v \in S_i$, find a mimimal minimum $r$-$v$ cut in the graph $G'$ defined above after shrinking subtrees incident to $S_i$.
Let $S_v$ denote this cut and let ${S'_v} = S_v \cap S_i$. We will say that $r$ is a {\em central vertex} for $S_i$ if
for each $v \in S_i, \ v \neq r$, $|{S'_v}| \leq |S_i|/2$. Let us first show that such a vertex exists.

\begin{lemma}
For any partition $S_i$, a central vertex $r$ exists for $S_i$.
\end{lemma}

\begin{proof}
Let $T$ be the eventual Gomory-Hu tree found by the sequential algorithm stated above. 
Remove all vertices not in $S_i$ from $T$. The resulting graph, say $T'$, will still be connected, since the finer partitions of $S_i$ 
always form a connected subtree of the tree on partitions at any stage of the algorithm. 
It is easy to see that there is a vertex $r \in T'$ such that each subtree of $T'$ incident at $r$ has at most $|T'|/2$ vertices.
Since $T$ is a Gomory-Hu tree, for each $v \in S_i, \ v \neq r$, a minimum $v$-$r$ cut is defined by one of the edges of 
$T$ that lies in $T'$. It follows that each such cut satisfies $|{S'_v}| \leq |S_i|/2$ and hence $r$ is a central vertex for $S_i$.
\end{proof}

A central vertex for $S_i$ can be found in \NC: For each vertex $r \in S_i$, test if it is a central vertex by finding, in parallel, a
minimal minimum $v$-$r$ in $G'$ for each vertex $v \in S_i, \ v \neq r$. From now on, let $r$ denote a central vertex for $S_i$.
The following fact is straightforward:    

    \begin{lemma}
    Let $r, u, v \in V$
    and let $S_u$ and $S_v$ be minimal minimum $u$-$r$ and $v$-$r$ cuts in $G$, respectively. Then $S_u$ and $S_v$ do not cross.
    \end{lemma}

\begin{corollary}
\label{cor.laminar}
Let $r, v_1, \ldots v_k \in V$
    and let $S_{v_1}, \ldots, S_{v_k}$ be minimal minimum $v_1$-$r$, ... $v_k$-$r$ cuts in $G$, respectively. Then $S_{v_1}, \ldots, S_{v_k}$
form a laminar family.
    \end{corollary}

Let $r$ denote a central vertex for $S_i$ that is found by the algorithm. By \cref{cor.laminar}, the cuts $S_v$, for each vertex $v \in S_i, \ v \neq r$
form a laminar family. Let $M_1, \ldots , M_l$ be the maximal sets of this laminar family. Clearly, we can split $S_i$ into the
$l$ sets $M_1 \cap S_i,  \ldots , M_l \cap S_i$ and attach subtrees to appropriate sets as given by $M_1, \ldots , M_l$.
This can be done for all sets $S_i$ of the current partition, in parallel. This defines one iteration of our parallel algorithm.
Clearly, after each iteration, the cardinality of the largest set in the partition drops by a factor of 2 and therefore only $O(\log n)$
such iterations are needed. Hence we get:

\begin{theorem}
\label{thm:gomoryhu}
There is an \NC~algorithm for obtaining a Gomory-Hu tree for an edge-weighted planar graph.
\end{theorem}

Now we use Padberg and Rao's theorem that states that the Gomory-Hu tree of a graph must contain a minimum odd cut as one
of its edges \cite{PadbergR} to finish the proof of \cref{lem:minoddcut}.

\begin{proof}[Proof of \cref{lem:minoddcut}]
	We first find a Gomory-Hu tree, then try all of the cuts obtained by removing an edge of the tree. We return the minimum among cuts that split the vertices into odd pieces. Clearly all of this can be done in parallel, and hence the algorithm is in \NC.
\end{proof}

\begin{remark}
An alternative way of finding a minimum odd cut $S$ is to use the Pickard-Queyranne structure of minimum $s$-$t$ cuts
\cite{picard}. However, that method is more cumbersome to describe.
\end{remark}
	\section{Extensions}
\label{sec:extensions}

In this section we will build on the machinery established in the previous sections to prove two
generalizations of \cref{thm:main}.

\begin{theorem}
\label{thm:two}
	There is an \NC~algorithm which given an edge-weighted planar graph with polynomially bounded 
	weights, returns a minimum weight perfect matching in it.
\end{theorem}

\begin{proof}
For a fixed integer $k$, assume that we are given a weight function on the edges of planar graph 
$G = (V, E)$,
\[W: E \rightarrow \{0, \ldots , n^k \}, \] 
and we wish to find a minimum weight perfect matching in $G$.
Recall that in \cref{sec:avg}, for the purpose of finding a perfect matching in $G$,
we had defined $0/1$ weights on edges given by function $w$.
Now, define the following composite weight function, $c$; its least significant $\log n$ bits
correspond to $w$ and the rest of the bits correspond to $W$.
\[ c_e = (n \cdot W_e) + w_e .\]

Clearly $c$ is polynomially bounded and can be used in place of $w$ to carry out 
the \NC~algorithm given in the previous sections.
The algorithm will return a minimum weight perfect matching w.r.t. weights given by $c$.
Since $w$ is $0/1$, for any perfect matching, the sum of weights according to $w$ is at most $n/2$.
Therefore in computing the weight of a perfect matching according to $c$, there will be no carry 
over from the least significant $\log n$ bits to the rest.
Hence the minimum weight perfect matching w.r.t. $c$ will also be a 
minimum weight perfect matching w.r.t. $W$.
\end{proof}

\begin{theorem}
\label{thm:three}
	There is an \NC~algorithm which given a bounded-genus graph, 
	 returns a perfect matching in it, if it has one.
\end{theorem}

\begin{proof}
	In order to derive our algorithm for planar graphs, we used planarity in exactly three ways, and here we will show how one can obtain the same results for graphs embeddable on orientable surfaces of bounded genus.
	\begin{enumerate}
		\item $\Omega(\card{E})$ edge-disjoint even walks (\cref{lem:evenwalks}): We used planarity to extract $\Omega(\card{E})$ edge-disjoint faces and then argued that by pairing up the faces we get $\Omega(\card{E})$ edge-disjoint even walks.
		\item Counting perfect matchings (\cref{lem:avg}): We used planarity to argue that we can count the number of minimum weight perfect matchings and hence get a point inside a face of the matching polytope $\PM(G, w)$ when $w$ is polynomially bounded.
		\item Finding the minimum odd cut (using \cref{thm:gomoryhu}): We used the fact that minimal minimum $s$-$t$ cuts can be computed in \NC~for planar graphs \cite{johnson}, in order to prove that we can construct Gomory-Hu trees and find the minimum odd cut.
	\end{enumerate}
	
	Note that given a graph, one can find an embedding onto a surface of genus $g=O(1)$ in \NC~if one exists \cite{EK14}. So from now on, we assume this embedding is given to us. We now address how each of \cref{lem:evenwalks,lem:avg,thm:gomoryhu} can be proved for bounded genus graphs.

	For \cref{lem:evenwalks}, note that we simply need to obtain $\Omega(\card{E})$ edge-disjoint cycles in our graph. Pairing up odd cycles can be done as before using a spanning tree. In planar graphs these cycles were obtained from the faces, and we used Euler's formula $\card{V}-\card{E}+\card{F}=2$ to argue that in graphs without degree $1$ vertices and with at most half of the vertices having degree $2$, there must be $\Omega(\card{E})$ faces. We still have an Euler's formula in the case of bounded genus graphs, but with $2$ replaced by a (negative) constant. The proof still works as before and one can show that $\card{F}\geq a\card{E}-b$ for some $a>0$, which implies that $\card{F}=\Omega(\card{E})$.
	
	For \cref{lem:avg}, it is enough to be able to count matchings. To be more precise, given weights $w$ over the edges of the graph, we simply need to compute the perfect matching generating function
	\[ \sum_{M\text{ is perfect matching}} \prod_{e\in M}w_e, \]	
	in \NC~as long as the bit complexity of $w$ is polynomially bounded. Mahajan and Varadarajan showed how this can be done in \NC~by slightly modifying an algorithm of Gallucio and Loebl \cite{MahajanV, GL99}. Their method reduces computing the matching generating function to taking a linear combination of Pfaffians over planar graphs.
	
	Finally, for finding minimum odd cuts in bounded genus graphs: We will show that we can use the methods of Borradaile et al.~to find the minimum odd cut in \NC~\cite{BENW14}. The algorithm of Borradaile et al.~allows one to find minimum cuts between all pairs of vertices in bounded genus graphs in nearly linear time, however we will show that a slight modification of it runs in \NC. This almost shows that one can find the minimum odd cut in \NC, because every minimum odd cut is also a minimum $s$-$t$ cut for some $s$ and $t$. However, one still needs to be careful about cases where there can be multiple minimum $s$-$t$ cuts.
	
	The main idea behind the algorithm of Borradaile et al.~is that a minimum cut separating vertices $s$ and $t$ is composed of dual cycles (of which there are at most $2^{O(g)}$), all but one of which can be chosen from certain homology classes without regards to the pair $s$ and $t$. They use this observation to reduce the problem to finding minimum cuts in $2^{O(g^2)}$ planar graphs, where $g$ is the genus of the original graph. Roughly speaking, they enumerate all possible homology classes for all but one the cycles, and one by one, from each homology class they find the shortest possible cycle in the chosen homology class and perform some surgery on the graph and its embedding. These surgeries reduce the genus, until the embedding becomes planar. The surgeries are easy to perform in \NC, as long as the cycles are found in \NC.
	
	In order to define the homology classes, and also find shortest cycles given the homology class, Borradaile et al. use the results of Erikson and Nayyeri \cite{EN11}. Their algorithm again can be implemented in \NC. At a high level the ideas involved are construction of a spanning forest in the dual graph in order to define $\Z_2$-homology signatures, and then constructing a $\Z_2$-homology cover of the bounded genus graph and finding shortest paths between pairs of vertices in the cover. All of these can be parallelized and are in particular in \NC as long as the genus is bounded.
	
	We now point out the main technicality needed to adapt the algorithm of Borradaile et al. Although minimum cuts between all pairs of vertices can be found from the minimum cuts in the $2^{O(g^2)}$ planarizations, it is not guaranteed that these cuts are nicely uncrossed from each other and form a Gomory-Hu tree. In fact, Borradaile et al.~perturb the weights in order to have unique minimum cuts, and in order to be able to merge the Gomory-Hu trees from $2^{O(g^2)}$ planar graphs into one Gomory-Hu tree for the original graph. We cannot afford to perturb the weights however because we do not have access to random bits. Instead we argue that a minimum odd cut can be directly found in one of the planarizations. Therefore one can produce the planarizations in \NC~and then construct a Gomory-Hu tree from each and find the minimum odd cut amongst them.
	
	Note that if the weights of the graph were slightly perturbed, then the minimum odd cut would have been the {\em unique} minimum $s$-$t$ cut for some pair $s$ and $t$, and we would have been able to find the unique cut in one of the planarizations. But this shows that even if the edges were not perturbed, a minimum odd cut must survive the surgeries performed on the graph for some sequence of fixed homology signatures. In other words, a minimum odd cut must be comprised of dual cycles, all but one of which are the minimum cycles from given homology classes. So a minimum odd cut must survive one of the sequences of surgeries performed on the graph, and found at the end by our algorithm for finding Gomory-Hu trees in planar graphs.
\end{proof}

We remark that all of the subroutines mentioned in the previous proof still remain in \NC~for genus 
up to $O(\sqrt{\log n})$, except for the subroutine that finds the embedding of the graph. 
Hence, \cref{thm:three} can be slightly 
strengthened to handle graphs of genus $O(\sqrt{\log n})$, as long as the input graph is given along 
with its embedding. Finally observe that the common generalization of \cref{thm:two,thm:three} 
easily follows.

	\section{Discussion}

The main open problem of course is to go beyond bounded genus graphs and obtain an \NC~perfect matching algorithm for general, or even bipartite, graphs. Below we state some more easily accessible 
open problems.

We note that $K_{3,3}$-free graphs may have genus as high as $O(n)$. Counting the number of perfect matchings for this class of graphs is in \NC~\cite{K33}.
Can our algorithm be extended to obtain an \NC~algorithm for the search version? We note that an \NC~algorithm for finding an $s$-$t$ min-cut in $K_{3,3}$-free graphs 
would give this result.

An interesting problem defined by Papadimitriou and Yannakakis \cite{PY}, called Exact Matching, is the following: Given a graph $G$ with a subset of the edges marked red and
an integer $k$, find a perfect matching with exactly $k$ red edges. This problem is known to be in \RNC~\cite{MVV}, even though it is not
yet known to be in \P. For the case of planar graphs, the decision version of this problem is known to be in \NC, though the search version is not (it is easy to
check that the search version is in \P).
Can our techniques be used to obtain an \NC~algorithm for the search version?

\begin{remark}
	Very recent work \cite{EV} has resolved the open problem stated above about $K_{3,3}$-free graphs. \cite{EV} go further to give \NC~algorithms for finding a perfect matching,
	a minimum weight perfect matching if the weights are polynomially bounded, 
	and an $s$-$t$ min-cut in one-crossing-minor-free graphs. 
\end{remark}

	\section{Acknowledgements}

We wish to thank David Eppestein, László Lovász, and Satish Rao for valuable discussions.
	\printbibliography
\end{document}